\documentclass[aps,twocolumn,noshowpacs, showkeys, superscriptaddress, floatfix,nobalancelastpage, notitlepage]{revtex4-2}

\usepackage{graphicx}% Include figure files
\usepackage{dcolumn}% Align table columns on decimal point
\usepackage{bm}% bold math
\usepackage{quantikz}
\usepackage{amsmath,amssymb,mathtools} 
\usepackage{amsthm,epsfig,tikz}
\usepackage[justification=justified]{caption}
\usepackage{xcolor}
\usepackage{verbatim}
\usepackage[normalem]{ulem}
\usepackage[rightcaption]{sidecap}
\usepackage{subcaption}
\usepackage{forest,adjustbox}
\usepackage{multirow}
\usepackage{hyperref}  
\hypersetup{colorlinks,citecolor=black,urlcolor=black,linkcolor =black,hypertexnames=true,pdfproducer={Me}}
\usepackage[mode=buildnew]{standalone}
\usepackage{enumerate}

\usepackage{float}
\makeatletter
\let\newfloat\newfloat@ltx
\makeatother
\usepackage{algorithm}
\usepackage{algpseudocode}
\usepackage{xcolor}
\usepackage{framed}
\definecolor{shadecolor}{cmyk}{0,0,0,0.05}%change what ever you like
\usepackage{tabularx,booktabs}
\usetikzlibrary{arrows.meta,
                automata,
                positioning,
                quotes}

\algrenewcommand{\algorithmiccomment}[1]{\hskip2pt\textbackslash\textbackslash \ #1}

\def\bra#1{\mathinner{\langle{#1}|}}
\def\ket#1{\mathinner{|{#1}\rangle}}

\renewcommand{\part}[2]{\frac{\partial #1}{\partial #2}}

\newcommand{\minus}{  \scalebox{0.45}[1.0]{\( - \)}  }

\newcommand\define[1]{\emph{\textbf{#1}}}

\newtheorem{re}{Remark}

\newtheorem{claim}{Claim}

\newcommand\psm[1]{$\begin{psmallmatrix}#1\end{psmallmatrix}$}

\begin{document}

\title{QuOp: A Quantum Operator Representation for Nodes}

\author{Andrew Vlasic}
\email{avlasic@deloitte.com}
\affiliation{New Business Initiative, Deloitte Consulting LLP, Chicago, IL}

\author{Salvador Aguinaga}
\email{saguinaga@deloitte.com}
\affiliation{Advisory's AI Center of Excellence, Deloitte Advisory LLP, Mishawaka, IN}

\begin{abstract}
We derive an intuitive and novel method to represent nodes in a graph with special unitary operators, or quantum operators, which does not require parameter training and is competitive with classical methods on scoring similarity between nodes. This method opens up future possibilities to apply quantum algorithms for NLP or other applications that need to detect anomalies within a network structure. Specifically, this technique leverages the advantage of quantum computation, representing nodes in higher dimensional Hilbert spaces. To create the representations, the local topology around each node with a predetermined number of hops is calculated and the respective adjacency matrix is used to derive the Hamiltonian. While using the local topology of a node to derive a Hamiltonian is a natural extension of a graph into a quantum circuit, our method differs by not assuming the quantum operators in the representation a priori, but letting the adjacency matrix dictate the representation. As a consequence of this simplicity, the set of adjacency matrices of size $2^n \times 2^n$ generates a sub-vector space of the Lie algebra of the special unitary operators, $\mathfrak{su}(2^n)$. This sub-vector space in turn generates a subgroup of the Lie group of special unitary operators, $\mathrm{SU}(2^n)$. Applications of our quantum embedding method, in comparison with the classical algorithms GloVe (a natural language processing embedding method) and FastRP (a general graph embedding method, display superior performance in measuring similarity between nodes in graph structures.
%The method is benchmark against classical computing methods with open data sets.      
\end{abstract}

\date{\today}

\maketitle

\section{Introduction}\label{sec:intro}

This manuscript describes a node embedding technique where the representations are a quantum operator, or special unitary matrix, and leverages quantum computing to calculate the similarity score between nodes. The method does not require parameter training, thereby significantly decreasing time to extract information, and performs quite well on graphs with random or limited edge connections. 

Graphs, a branch of mathematics focusing on the study of vertices and edges (abstracting relationship between nodes), hold significant importance in the current landscape of artificial intelligence (or AI), with node/graph embedding the most widely applied \cite{makarov2021survey}.  Given that graphs can be described with matrices, there is a natural extension of techniques that extract the latent information embedded in a graph to quantum computation \cite{aharonov2001quantum,childs2010relationship,wang2020experimental,tanaka2022spatial,malmi2022spatial,goldsmith2023link,vlasic2023scoring}. However, to date, the discipline of quantum computing has yet to derive an analog of the vector representation that compresses and captures the latent information in a node. To clear any potential confusion, there is a distinction from the feature mappings techniques of graphs, which encode the structure of a graph into a quantum circuit for a downstream task \cite{suzuki1976generalized,calude2017qubo,garg2019quantum,verdon2019quantum,henry2021quantum,skolik2023equivariant}, and the representation of a node as an object, such as a vector.   

Taking inspiration from continuous quantum random walks over a graph \cite{aharonov2001quantum,burda2009localization,childs2010relationship}, we derive a method to embed nodes from a given graph to the space of special unitary matrices, or quantum operators, which we denote as \textbf{QuOp}. For an arbitrary node from a symmetric graph, the embedding consists of the local adjacency matrix around this node from a predetermined number of hops. The inner product of two embeddings yields a similarity score in the unit interval, with zero indicating no similarity and one an exact match. The inner products are calculated in a quantum circuit, and thus, represent nodes in a higher dimensional Hilbert space \cite{schuld2019quantum}. For directed acyclic graphs, we employ the so called Hermitian
adjacency matrix to transform the matrix to the proper form of a Hermitian. 

In the construction of the embedding the adjacency matrix serves as the Hermitian operator for the time-independent Schr{\"o}dinger's equation. For the construction, we set time $t=1$ and Planck's constant $\hbar=1$. The constants are fixed since the Hamiltonian evolution around the node is not considered, but only the exponential map from the Lie algebra to the respective Lie group. From this perspective, the collection of local adjacency matrices generates the Lie subalgebra, and hence, generates the Lie subgroup that determines the base gates that reflect respective circuits. Therefore, representing the nodes via the Lie subgroup of quantum operators. This representation to capture the latent information of a graph is best suited for random graphs, that is, a graph with a distribution of edges and distribution of weights on edges, as the topological structure is captured. 

This method is variational free, and therefore, requires no training for applications. Thus, the algorithm has immediate application, bypassing the lag-time for training. Furthermore, small graphs may not have enough information to train robust classical embeddings. However, this is not an issue with the proposed algorithm, as it is able to take advantage of the available information.  

Finally, the number of qubits scale logarithmically with respect to the length of the adjacency matrix. However, the number of gates scales polynomially to density of the graph.  

The organization of this manuscript is as follows. To establish a basis, Section \ref{sec:classical} gives an overview of graphs and graph embedding methods on a classical system. Section \ref{sec:QuOp} explicitly describes the QuOp, but first gives an overview of quantum continuous random walks to motivate the algorithm. Furthermore, within this section, QuOp is compared against similar graph encoding methods and benchmarking against classical computing embedding techniques. For the theoretical basis of QuOp, Section \ref{sec:lie} gives an overview of Lie algebras and Lie groups and then describes how the QuOp algorithm generates a Lie subalgebra and, as a consequence, the Lie subgroup of special unitary operators. Finally, Section \ref{sec:discussion} gives a brief summary of the algorithm and benchmarking, and a description of potential future research, including other methods to represent nodes as quantum operators.

\section{Classical Graph Embedding}\label{sec:classical}

Graph embeddings are extremely important to the field of AI, as they convert unprocessed data into a compact, numerical representation that captures relevant relationships~\cite{pennington2014glove}. These representations make it easier for machines to process and learn from intricate data like text, audio, and images~\cite{mikolov2013efficient}. Particularly, by using embeddings, AI systems can more effectively make predictions, classify, cluster, make recommendations, and understand natural language.

\paragraph{Graph Embeddings} Node embeddings have revolutionized the field graphs and AI through graph-based representation learning. Two of the most influential works in the field include DeepWalk by Perozzi et al.~\cite{perozzi2014deepwalk}, and Node2Vec by Grover and Leskovec \cite{grover2016node2vec}. Both works have been extensively cited and have paved the way for significant advancements in AI.
Perozzi's work introduced an innovative technique for the online learning of latent social representations. This work has been instrumental in developing AI methodologies, particularly in social network analysis, by enabling the use of local information obtained from truncated random walks to learn latent representations. It should be noted that in our experiments, we examine GloVe embeddings, which is a word embedding technique using global co-occurrence statistics to learn vector representaions
~\cite{pennington2014glove}.

On the other hand, Node2Vec~\cite{grover2016node2vec} presented a semi-supervised algorithm for generating node embeddings flexibly and efficiently. The model's balancing of local and global network structures~\cite{wills2020metrics} has been acknowledged for its significance in AI, especially in applications such as link prediction, clustering, and classification. Both works have considerably influenced the direction of research and development, particularly in utilizing node embeddings for network analysis. Okuno et al. \cite{okuno2019graph} proposed a novel graph embedding technique that enhances the ability of graph kernels to estimate graph edit distance. This has led to a wider range of applications of graph embeddings, particularly in tasks that involve the computation of graph similarity~\cite{okuno2019graph,vishwanathan2010graph}.

\paragraph{Text Embeddings} Recent developments in text embeddings achieved significant advancements with the arrival of transformer-based models. Transformer architectures, particularly BERT (Devlin et al. \cite{devlin2018bert}) and its variants, have revolutionized natural language processing by capturing contextual information effectively. The attention mechanism and pre-training strategies in these models have led to improved embeddings, fostering breakthroughs in tasks such as sentiment analysis, named entity recognition, and question answering. Notable works include RoBERTa~\cite{liu2019roberta} and ALBERT~\cite{lan2019albert}, which have refined BERT's architecture, enhancing efficiency and performance in various text-related applications.

\paragraph{Document Embeddings} In the realm of graph theory applications to text, notable advancements in graph embeddings have emerged, leveraging the inherent structure of textual data. Models like GraphSAGE, Hamilton et al. \cite{hamilton2017inductive}, and Graph Attention Networks, Velickovic et al. \cite{casanova2018graph}, have enabled the creation of meaningful representations for documents, capturing relationships between words and sentences. 

A graph $G$ is typically denoted as the pair $(V,E)$ where $V$ is the set of vertices and $E$ is the set edges between connected vertices. Given that two nodes are connected or not is binary, $E$ can be represented as a matrix. For $|V|=n$, each node is given a unique label $1,\ldots,n$. Then we can construct a matrix $A$ that represents $E$ where entry $a_{ij}= \left\{\begin{array}{cc}
   1  & \mbox{if } e(i,j)\in E \\
   0  & \mbox{otherwise}
\end{array} \right. .$  The matrix $A$ is denoted as the \define{adjacency} matrix. 

In general, there is a respective function $w(i,j)$ where $a_{ij} = w(i,j)$. Ergo, the output is the weight on the edge connecting node $i$ to node $j$, and if no edge exists then the value is $0$. If the graph is not weighted then this value is $1$, and if the graph is symmetric then $w(i,j) = w(j,i)$ for all nodes $i$ and $j$. This is the basic graph structure that is utilized to train node embeddings. While this representation captures the structure information of a graph, there is latent information that describes the underlying flow. It is this latent information that we aim to extract.

To extract the latent information within a node $v_i$, the node is ``flattened'' to a vector of dimension $n$, and $n$ is chosen a priori. The general mathematical basis is given through the probability of observing a neighboring node $v_j$ given the anchor node $v_i$. The probability measure is assumed to be a Gibbs measure, which is Markovian and convex \cite{georgii2011gibbs}. The probability measure is represented as follows, see Figure~\ref{fig:node-edge-node} and refer to Zhou et al.~\cite{zhou2023co} for context in general graphs, 
\begin{equation*}\label{eq:node2vec-embedding}
\displaystyle Pr(v_j\mid v_j) = \frac{exp \left( \Theta(v_i)^T \cdot \Theta(v_j) \right)}
                  {\sum_{k \in N(v_i)} {exp\left( \Theta(v_j)^T \cdot \Theta(v_k) \right)}} .
\end{equation*}

\begin{figure}[h!]
    \centering
    %\includestandalone[scale=1.3]{figures/nodes-attribs}
        \begin{tikzpicture}[auto,
     node distance = 15mm,
every state/.style = {fill=blue!60,text=white, 
                      minimum size=2em, inner sep=1pt, outer sep=1pt},
                > = Stealth,
every edge/.style = {draw, ->}
                        ]
% state nodes
\node (vi) [state, accepting,above, "\psm{\mathbf{x_1}\\\mathbf{x_2}\\\dots}"]   {$v_i$};
\node (vj) [state, above, "\psm{\mathbf{x_1}\\\mathbf{x_2}\\\dots}", right=of vi] {$v_j$};
% connections
\path
    (vi)    edge []  (vj)
    (vi)    edge ["\psm{\mathbf{e_1}\\\mathbf{e_2}\\\dots}"]              (vj);
    \end{tikzpicture}
    \caption{Graph schema showing anchor node $v_i$ and neighboring node $v_j$ connected by an edge. Each node and edge can have one or more attributes.}
    \label{fig:node-edge-node}
\end{figure}
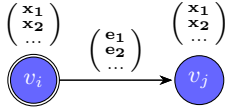

Within Figure~\ref{fig:node-edge-node}, $\Theta(v_i)$ is the vector representation for node $v_i$, $\Theta(v_i)^T \cdot \Theta(v_j)$ is the dot product of the embedded nodes, and $\Theta( \cdot )$ is the mapping derived from the neural network where $\Theta: V \to \mathbb{R}^d$ for a predetermined $d\in \mathbb{N}$. For the local topology to learn form, $N(v_i)$ is denoted as the neighborhood around the vertex $v_i$. Exactly as in Gibbs measure, the denominator is a scaling term, ensuring the probabilities sum to $1$ over all neighboring nodes. This term is the sum of the exponential dot products of the embedding of anchor node with its neighbors.  

Training embeddings is a fairly straightforward process with the natural objective function
\begin{equation*}
    \begin{split}
         \min_{\Theta} & \log\left( \sum_{k \in N(v_i)} exp\left( \Theta(v_j)^T \cdot \Theta(v_k) \right) \right) \\
        & - \log\left( \Theta(v_i)^T \cdot \Theta(v_j) \right),
    \end{split}
\end{equation*}
adjusting accordingly to the graph data and neural architecture. For brevity, a deeper dive is left to the reader.

Over all, graph embedding techniques converts intricate topological structure and relationships within a graph into a compact numerical vector representation. At the core, the objective is to map nodes, edges, or the entire graph into a lower-dimensional vector space while preserving structural and attribute information. A node and or edges often include attributes which reflect characteristics local to the node or node pair. The relationships between a pair of nodes bring together information from neighboring nodes and the neighbors of their neighbors. The rich information mapped to vector space is much easier to apply machine learning models directly~\cite{makarov2021survey}.

\section{QuOp Algorithm}\label{sec:QuOp}

\subsection{Motivation and Description}\label{subsec:motivation}

To motivate the algorithm, we briefly describe the procedure of conducting a continuous quantum random walk over graph \cite{aharonov2001quantum,childs2002example,childs2010relationship,wang2020experimental,malmi2022spatial,tanaka2022spatial,goldsmith2023link,vlasic2023scoring}. 

For a symmetric graph $G$, the continuous random walk requires a time-step $\gamma$ and Hermitian matrix $B$. Then for integer $t$ (required for the time-step) define the Hamiltonian $H_B(\gamma,t) := \exp\left\{ \minus i\gamma t \cdot B \right\}$, which follows from Schr{\"o}dinger's equation where the matrix is independent of time; throughout the manuscript we set Planck's constant $\hbar=1$. Typically, $B$ is the discrete Laplacian of the graph or the adjacency matrix, but may be in other forms~\cite{burda2009localization}, as long as the matrix captures the interactions of the nodes within the graph and the matrix is Hermitian \cite{nielsen2001quantum}. 

\begin{re}\label{re:padding}
There is potential that the size of the Hermitian matrix $B$ is not of a power of $2$, which is required to implement the operator in a quantum circuit. To mitigate the dimension issue, matrices are padded with zeros, where the padding adjusts the matrix to size $2^m \times 2^m$, for $m$ that holds the inequality $\displaystyle 2^{m-1}< \max_{ e\in E } \sqrt{ \mbox{dim}(A^n_e) } \leq 2^m$.
\end{re}

For continuous random walks, observe that the number of qubits required are $n$ for a matrix of size $2^n \times 2^n$. Ergo, the length of the quantum circuit grows logarithmically. Furthermore, the depth of the circuit is dependent on the density of the graph with respect to the edges, and contingent on the native gates of the quantum processor. However, we may state that the growth of the circuit depth will grow at least polynomially.   

\begin{re}\label{re:permutations}
While the graph structure is fixed, the labeling of the nodes to a column/row in the adjacency matrix is open, with many permutations. However, the labeling of the nodes does not affect the dynamics since there exists a permutation matrix $\mathcal{P}$ where, for two adjacency matrices $A$ and $\Tilde{A}$, $A = \mathcal{P}^T \Tilde{A} \mathcal{P}$. 
For more in-depth information for graphs in general, see Bapat \cite{bapat2010graphs}.
\end{re}

For applications, sequential implementations of the operator $H_B(\gamma,t)$ converges to a probability that is contingent on the initial condition. Further discussion is out of scope of the paper; see \cite{aharonov2001quantum,vlasic2023scoring} for in-depth information. 

The QuOp algorithm, with each node $v_i$ and a predetermined number of hops, will use the operator $H_{ B_{v_i} }(\gamma=1,t=1)$, where $B_{v_i}$ is the local adjacency matrix around node $v_i$. Since the algorithm is not concerned with the evolution of the Schr{\"o}dinger's equation, $\gamma$ and $t$ are fixed. Therefore, we take the parameters to be the multiplicative identity. The reason for the choice of the adjacency matrix is thoroughly described in Section \ref{sec:lie}. 

If the graph $G$ is not symmetric, a transformation to the adjacency matrix, called the Hermitian-adjacency \cite{guo2017hermitian}, is applied. For generality, we take the derivation given by Mohar \cite{mohar2020new} for the transformation. In particular, for $\alpha = a + bi$ where $|\alpha|=1$ and $a \geq 0$, the Hermitian adjacency matrix, denoted as $A^{\alpha}$, has entries of the form 
\begin{eqnarray}\label{eq:herm-adj}
a^{\alpha}_{ij} = \left\{
\begin{array}{ll}
      a_{ij}   & \mbox{if } a_{ij} = a_{ji}  \\
       a_{ij}\alpha + a_{ji}\overline{\alpha}  & \mbox{otherwise}
\end{array} 
\right.
\end{eqnarray}
Define the Hermitian adjacency transform function as $\mathrm{hermAdj}( A, \alpha )$. This transformation, while convenient to ensure the adjacency matrix is Hermitian, lifts the real valued dynamic to the complex plane. Thus, capturing the intricate dynamic of the graph in a higher dimensional Hilbert space.

\begin{figure}[!ht]
    \centering
    \begin{subfigure}[b]{.5\textwidth}
        \centering
        \includegraphics[width=260px]{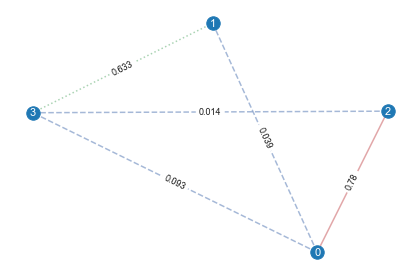}
        \caption{Simple symmetric weighted graph with weights strictly in the unit interval.}
        \end{subfigure}  
    % \begin{subfigure}[b]{.5\textwidth}
    %   \centering 
    %   \begin{tikzpicture}
    %   [scale=.8,auto=left,every node/.style={circle,fill=blue!20}]
    %   \node (n0) at (5,-2) {0};
    %   \node (n1) at (3,2)  {1};
    %   \node (n2) at (6,0)  {2};
    %   \node (n3) at (-2,0)  {3};
      
    %   \foreach \from/\to in {n3/n1,n1/n0,n3/n2,n3/n0,n3/n0,n0/n2}
    %     \draw (\from) -- (\to);

    % \end{tikzpicture}\caption{Simple symmetric weighted graph with weights strictly in the unit interval.}
    % \end{subfigure}
    \begin{subfigure}[b]{.5\textwidth}
    \centering
    \scalebox{.7}{
    \begin{quantikz}[thin lines] 
        \lstick{$\ket{0}$}  &  \gate{R_x(0.053) } & \gate[2]{R_{xx}(0.093)} & \gate[2]{R_{zx}(0.147)} & \gate[2]{R_{yy}(\minus 0.093)} & \gate[2]{R_{zx}(0.025)} 
       \\ \lstick{$\ket{0}$}  & \gate{R_x(1.413)}  & \qw   & \qw  & \qw & \qw 
    \end{quantikz}
     } \caption{Circuit for node $0$.}
     \end{subfigure}  
     \begin{subfigure}[b]{.5\textwidth}
    \centering
    \scalebox{.7}{
    \begin{quantikz}[thin lines] 
        \lstick{$\ket{0}$}  &  \gate{R_x(0.039) } & \gate[2]{R_{xx}(0.633)} & \gate[2]{R_{zx}(0.093)} & \gate[2]{R_{yy}(0.633)} & \gate[2]{R_{zx}(0.0039)} 
       \\ \lstick{$\ket{0}$}  & \gate{R_x(0.093)}  & \qw   & \qw  & \qw & \qw 
    \end{quantikz}
    } \caption{Circuit for node $1$.}
     \end{subfigure}

     \begin{subfigure}[b]{.5\textwidth}
    \centering
    \scalebox{.7}{
    \begin{quantikz}[thin lines] 
        \lstick{$\ket{0}$}  &  \gate{R_x(0.780) } & \gate[2]{R_{xx}(0.014)} & \gate[2]{R_{zx}(0.093)} & \gate[2]{R_{yy}(0.014)} & \gate[2]{R_{zx}(0.780} 
       \\ \lstick{$\ket{0}$}  & \gate{R_x(0.093)}  & \qw   & \qw  & \qw & \qw 
    \end{quantikz}
    } \caption{Circuit for node $2$.}
     \end{subfigure}

     \begin{subfigure}[b]{.5\textwidth}
    \centering
    \scalebox{.7}{
    \begin{quantikz}[thin lines] 
        \lstick{$\ket{0}$}  &  \gate{R_x(0.053) } & \gate[2]{R_{xx}(0.093)} & \gate[2]{R_{zx}(0.147)} & \gate[2]{R_{yy}(\minus 0.093)} & \gate[2]{R_{zx}(0.025} 
       \\ \lstick{$\ket{0}$}  & \gate{R_x(1.413)}  & \qw   & \qw  & \qw & \qw 
    \end{quantikz}
    } \caption{Circuit for node $3$.}
     \end{subfigure}
\caption{A simple graph with four nodes and respective circuits for operators representing each node are displayed. The simplicity of the edges between the nodes yields the same circuit architecture, and the different weights yield different parameters. The circuits were implemented and decomposed with Qiskit \cite{aleksandrowicz2019qiskit}, and the graph and adjacency matrices were calculated with NetworkX \cite{hagberg2008exploring}. }
\label{fig:node_circs}
\end{figure}

Making the description above more tangible, consider the toy graph in Figure \ref{fig:node_circs}. While extremely simple, with only four nodes, the special unitary operators created from the one hop adjacency matrix which represent each node are still fairly intricate; Figure \ref{fig:node_circs} displays the graph and operators representing each node. The simplicity of the graph is the reason why the gates are the same for the operator representations for each node. With larger graphs, such behavior should not be expected as the sparsity of the adjacency matrices will significantly vary. 

To compare the similarity of the node operator representations there are a few kernel methods to make this calculation. To ensure the algorithm is further amenable for direct implementation, the descriptions of the \textit{fidelity test} and \textit{SWAP test} are given; these two kernel methods, and general kernel methods, are described quite well in the Pennylane documentation \cite{bergholm2018pennylane}. The outline of the quantum circuits for both methods are displayed in Figure \ref{fig:inner-prod-circuits}. 

The fidelity test requires a small number of qubits, but requires a lot of gates in the circuit. The SWAP test requires far less gates, but twice the number of qubits as the fidelity test plus one more qubit for the ancillary register. Furthermore, the SWAP test only measures one qubit and fidelity test measures all qubits in the circuit. For implementation, it may be necessary to weigh the respective strengthens and weaknesses of each method.

 For the pseudo-algorithm we will note the fidelity test.  For brevity, this circuit denoted as $\mathrm{fidelity}(U,U')$ for two special unitary operators $U$ and $U'$.

\begin{figure}[!ht]
     \centering
     \begin{subfigure}[b]{.5\textwidth}
    \centering
    \scalebox{1.}{
        \begin{quantikz}[thin lines] 
           \lstick{$\ket{0}^{\otimes m}$} & \qwbundle[]{m} &  \gate{ \exp\{\minus i A^h_i\}} & \gate{ \exp\{\minus i A^h_j\}^{\dagger} } & \meter{}
        \end{quantikz}
    } \caption{\label{fig:fidelity} This is an illustration of the fidelity test.}
     \end{subfigure}
     \begin{subfigure}[b]{.5\textwidth}
    \centering
    \scalebox{1.}{
    \begin{quantikz}[thin lines] 
          \lstick{$\ket{0}$}             & \qw            & \gate{H}                        & \ctrl{1} & \qw & \gate{H} & \meter{} 
       \\ \lstick{$\ket{0}^{\otimes m}$} & \qwbundle[]{m} &  \gate{ \exp\{\minus i A^h_i\}} & \swap{1} & \qw & \qw      & \qw
       \\ \lstick{$\ket{0}^{\otimes m}$} & \qwbundle[]{m} & \gate{ \exp\{\minus i A^h_j\}}  & \swap{0} & \qw & \qw      & \qw 
    \end{quantikz}
    } \caption{This is a general illustration of the SWAP test.}
    \label{fig:swap}
    \end{subfigure}
\caption{The local adjacency matrices were calculated with the local topology around the nodes of $i$ and $j$ with $h$ hops, and the matrices are the same size. Each circuit is measured in the computational basis. }
\label{fig:inner-prod-circuits}
\end{figure}
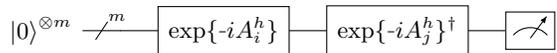
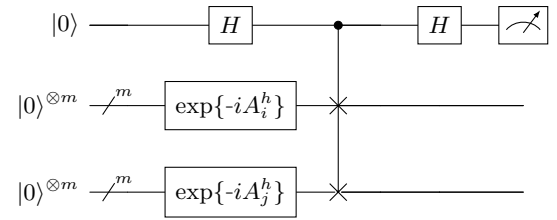

To implement the SWAP test, note with this calculation that if the states produced by the operators are orthogonal then the $0$-state is measured with a probability of $\frac{1}{2}$, and if the states produced by the operators are equivalent then the $0$-state is measured with a probability of $1$. Hence, a linear adjustment will map values to the unit interval.  

All of the necessary information has been discussed in order to describe the algorithm in detail, and the description is given in Algorithm \ref{alg:QuOp}. The illustration is very explicit and is intended for ease of implementation, and is not intended to be concrete or comprehensive. Therefore, the algorithm should be considered as a pseudo-algorithm, where individual implementations should be accordingly adjusted.  

\begin{algorithm}
\textbf{Input:} A graph $G$, set of nodes $\{ n_0,\ldots, n_i \}$, number of hops $h$, and Hermitian adjacency parameter $\alpha+i\beta$.
\begin{algorithmic}[1]
\State{$maxDim \gets 0$}
\For{$k \gets 0, i$}
    \State{$A^h_k \gets$ adjacency of hops $h$ from $G$}
    \State{$maxDim \gets \max\{ maxDim, \mathrm{dim}(A^h_k) \}$  }
    \State{$A^h_k \gets \mathrm{hermAdj}( A^h_k, \alpha )$ }
\EndFor
\State{$\mathrm{padsize} \gets 2^{ \lceil\log_2(maxDim) \rceil }$  \Comment{See Remark \ref{re:padding} }}
\For{$k \gets 0, i$}
    \State{$A^h_k \gets$ padding of zeros of size $(\mathrm{padsize}  \minus \mathrm{dim}(A^h_k))$ } 
    \State{$U_k \gets \exp\{ \minus i \cdot A^h_k \}$}
\EndFor
\State{$collectResults \gets$ array of zeros of size $(i+1) \times (i+1)$}
\For{$k \gets 0, i-1$}
    \For{$j \gets k+1, i$}
    \State{$collectResults[k][j] \gets \mathrm{fidelity}(U_k,U_j)$ }
    \EndFor 
\EndFor 
\State{\textbf{return} $collectResults$ }
\end{algorithmic}
\caption{QuOp Embedding}\label{alg:QuOp}
\end{algorithm}

The derivation of the algorithm opens the question of labeling nodes and the potential effect of how different labeling may have on the outcome. By Remark \ref{re:permutations}, for a node and its adjacency matrix from $h$ hops, different node labels change adjacency matrices only by a permutation matrix. Hence, while adjacency matrices from different labels create different circuits, for two nodes the respective gates should be the same and differ through a permutation on the wires in a circuit. Therefore, the expectations are the same. This logic is displayed in the claim below. The proof focuses on the fidelity test, however, the argument is easily extended to other algorithms where the expectation is an additional constant or multiple of the inner product. 

\begin{claim}
If adjacency matrices differ by a permutation matrix then the expectations in Algorithm \ref{alg:QuOp} are equivalent. 
\end{claim}
\begin{proof}
Take two arbitrary nodes, $k$ and $l$, take two arbitrary node labeling and respective adjacency matrices, $A$ and $\Tilde{A}$. Then for $h$ hops, there exists a permutation matrix $P$ such that $A^h_k = P^T \Tilde{A}^h_k P$ and $A^h_l = P^T \Tilde{A}^h_l P$. Recall that permutation matrices are real valued and orthogonal, hence $P^{-1} = P^{T}$, $P$ preserves the inner product, and $P$ has a finite order.  

Following Schuld and Petruccione \cite{schuld2021quantum}, for the fidelity test and observable $\mathcal{M}$ we see that 
\begin{equation*}
    \begin{split}
     & \bra{0}^{\otimes n} (e^{-i\cdot A^h_k})^{\dagger} e^{-i\cdot A^h_l} \mathcal{M} (e^{-i\cdot A^h_l})^{\dagger} e^{-i\cdot A^h_k} \ket{0}^{\otimes n} \\
     = & \bra{0}^{\otimes n} (e^{-i\cdot A^h_k})^{\dagger} e^{-i\cdot A^h_l} \ket{0}^{\otimes n}\bra{0}^{\otimes n} (e^{-i\cdot A^h_l})^{\dagger} e^{-i\cdot A^h_k} \ket{0}^{\otimes n} \\
     =& |\bra{0}^{\otimes n} (e^{-i\cdot A^h_l})^{\dagger} e^{-i\cdot A^h_k} \ket{0}^{\otimes n}|^2 \\
     =& |\bra{0}^{\otimes n} (e^{-i\cdot P^T \Tilde{A}^h_l P})^{\dagger} e^{-i\cdot P^T \Tilde{A}^h_k P} \ket{0}^{\otimes n}|^2 \\
     =& |\bra{0}^{\otimes n} (P^T e^{-i\cdot \Tilde{A}^h_l } P)^{\dagger} P^T  e^{-i\cdot \Tilde{A}^h_k}  P \ket{0}^{\otimes n}|^2 \\
     =& |\bra{0}^{\otimes n} P^T (e^{-i\cdot \Tilde{A}^h_l } )^{\dagger} e^{-i\cdot \Tilde{A}^h_k} P \ket{0}^{\otimes n}|^2 . 
    \end{split}
\end{equation*}
Now observe that 
\begin{equation*}
    \begin{split}
     & |\bra{0}^{\otimes n} (e^{-i\cdot \Tilde{A}^h_l})^{\dagger} e^{-i\cdot \Tilde{A}^h_k} \ket{0}^{\otimes n}|^2 \\
     =& |\bra{0}^{\otimes n}P^T P(e^{-i\cdot \Tilde{A}^h_l})^{\dagger} P^T P e^{-i\cdot \Tilde{A}^h_k}P^T P \ket{0}^{\otimes n}|^2 \\
     =& |\bra{0}^{\otimes n}P^T (e^{-i\cdot A^h_l} )^{\dagger}  e^{-i\cdot A^h_k}P\ket{0}^{\otimes n}|^2 \\
    \end{split}
\end{equation*}
Thus, if $P$ is of an odd order then repeated substitution equalities yield 
\begin{equation*}
    \begin{split}
    & |\bra{0}^{\otimes n} (e^{-i\cdot A^h_l})^{\dagger} e^{-i\cdot A^h_k} \ket{0}^{\otimes n}|^2 \\
    & = |\bra{0}^{\otimes n} (e^{-i\cdot \Tilde{A}^h_l})^{\dagger} e^{-i\cdot \Tilde{A}^h_k} \ket{0}^{\otimes n}|^2.
    \end{split}
\end{equation*}

If $P$ has an even order, then through repeated substitutions we can derive another permutation matrix. Ergo, we can find adjacency matrices $_1 A^h_j, \ldots, _d A^h_j$, where $j= l,k$, and respective permutation matrices $P_1,\ldots, P_d$ that gives the equality

\begin{equation*}
    \begin{split}
    & |\bra{0}^{\otimes n} (e^{-i\cdot A^h_l})^{\dagger} e^{-i\cdot A^h_k} \ket{0}^{\otimes n}|^2 \\
    & = |\bra{0}^{\otimes n} (P P_1 \ldots P_d)^T (e^{-i\cdot \Tilde{A}^h_l})^{\dagger} e^{-i\cdot \Tilde{A}^h_k} P P_1 \ldots P_d \ket{0}^{\otimes n}|^2.
    \end{split}
\end{equation*}
Following the same logic as above yields the proof. 

% Thus 
% \begin{equation*}
%     \begin{split}
% & |\bra{0}^{\otimes n} P^T (e^{-i\cdot \Tilde{A}^h_l } )^{\dagger} e^{-i\cdot \Tilde{A}^h_k} P \ket{0}^{\otimes n}|^2 \\
% = & |\bra{0}^{\otimes n} P (e^{-i\cdot \Tilde{A}^h_l } )^{\dagger} e^{-i\cdot \Tilde{A}^h_k} P^T \ket{0}^{\otimes n}|^2 
%     \end{split}
% \end{equation*}
% and since the permutation does not matter in the expectation the claim follows. 

\end{proof}

\subsection{Comparison with Other Methods}\label{subsec:compare}
The concept of representing a graph, or subgraph, with a Hamiltonian is quite natural, as the nodes themselves are the singular terms and edge between of nodes generates the quadratic terms \cite{calude2017qubo,cong2019quantum,verdon2019quantum,henry2021quantum,9978396,albrecht2023quantum,skolik2023equivariant}. For an illuminating example, determination of the existence of an isomorphisms between graphs are mapped to a quadratic unconstrained binary optimization (QUBO) problem \cite{calude2017qubo}, where the topology of the graphs are mapped as quadratic factors, and the lowest energy indicates how close in structure are the respective graphs. In fact, QUBOs have an exact mapping to an Ising Hamiltonian~\cite{farhi2000quantum}.  

There is a rich literature on encoding a graph into a quantum circuit, all of which have respective ansatz layers for a tunable Hamiltonian \cite{cong2019quantum,verdon2019quantum,henry2021quantum,9978396,albrecht2023quantum,skolik2023equivariant}, and all of which are designed for a downstream task in a variational algorithm. Henry et al. \cite{henry2021quantum} and Albrecht et al. \cite{albrecht2023quantum} designed an encoding map for graphs with both leveraging the Ising Hamiltonian and the $XY$ Hamiltonian. Interestingly, both authors consider a classification task and, since the neutral atom processors were used for the experiments, the Hamiltonian evolved over time. Repeated application of the layers approximates this evolution. 

With the consideration of time evolution, Skolik et al. ~\cite{skolik2023equivariant} derive a rigorous equivariant method against the node permutations which, as noted in the paper, is a special case of the quantum approximate optimization algorithm (QAOA). The ansatz is tailored to the specific task. 

While the QuOp algorithm was motivated from the Hamiltonian of continuous walks, this is where the similarities with the cited graph encoding methods end. For instance, the Hermitian operators are predefined for all of the previous mentioned techniques, and for QuOp the adjacency matrices generate the quantum operators and the respective sequence of gates. For the generated quantum gates, the Cartan decomposition may decompose the node embedding into a well-known Hamiltonian. Moreover, the QuOp algorithm is variational free, where the other techniques have variational layers trained on specific classification tasks. Finally, in Section \ref{sec:lie}, we describe the representation of the technique with the natural Lie algebras by leveraging the local topologies around each node.       

\subsection{Experiments}

\begin{table}
    \centering 
    \footnotesize
    \caption{Pairwise similarity in random graph. QuOp and FastRP scores for arbitrary node pairs.}
    \begin{tikzpicture}
    \node (table) [inner sep=1pt] {
    \renewcommand{\arraystretch}{1.1}%
    \setlength{\tabcolsep}{4pt}

    \begin{tabularx}{.97\linewidth}{llXX}
    &&\multicolumn{2}{c}{\bfseries{Number of Nodes}}\\
    \bfseries{\underline{Graph Node Size}} & \bfseries{\underline{Node Pairs}} & \bfseries{\underline{FastRP}} & \bfseries{\underline{QuOp}} \\%& \bfseries{\underline{128}}\\ 
    \multirow{4}{*}{32} 
    & (4,8)   &  .466    & .766 \\%& 1\\
    & (19,20) &  .203    & .827 \\%& 0.483\\
    & (7,29)  &  .353    & .157 \\%& 0.333\\
    & (9,30)  &  .188    & .207 \\%& 0.307\\\hline
    \hline\multirow{4}{*}{64}
    & (19,22) &  .392    & .855 \\%& -\\
    & (9,30)  &  -.071   & .816 \\%& -\\
    & (41,55) &  .490    & .038 \\%& -\\
    & (10,11) &  .390    & .141 \\%& -\\
    
    \end{tabularx}
    
    };
    \draw [rounded corners=.5em] (table.north west) rectangle (table.south east);
    \end{tikzpicture}
    \label{table:info-distance-fastrp}
    
\end{table}

Classical graph node embeddings play a pivotal role in network analysis, offering a means to represent complex graph structures in a lower-dimensional space while preserving important topological information. One approach to achieve this is through fast random projection (FastRP) methods, where nodes are projected onto a lower-dimensional space using random linear transformations~\cite{chen2019fast}. These methods capture structural characteristics of the graph by mapping nodes in such a way that preserves their pairwise relationships. We use a well-known, tried and tested graph database tool from Neo4J~\cite{neo4jfastrp} to store, retrieve, and run graph data science algorithms on the graph. In this case we use their FastRP implementation of random projections. 

Fast random projection methods, such as those leveraging techniques like random projection matrices, offer computational efficiency and scalability. These methods balance preserving graph topology and reducing the dimensionality of the node representations. The resulting embeddings facilitate downstream tasks like link prediction, community detection, and classification, as they capture the essential structural patterns within a graph.

Moving forward in the evolution of classical graph embeddings, we encounter the influential Node2Vec algorithm~\cite{grover2016node2vec}. Node2Vec is a state-of-the-art technique beyond mere structural similarity, incorporating notions of homophily and structural equivalence. By employing a biased random walk strategy, Node2Vec explores local and global neighborhoods, allowing nodes with similar roles or functions to be embedded closely in the vector space. Given the computational shortcomings of implementing QuOp, as discussed below, we will compare QuOp against GloVe~\cite{pennington2014glove}, which lies under this class of algorithms.

To test the efficacy of QuOp, we compare the performance of the method against FastRP applied to two randomly generated weighted graphs, as well as the well-used Karate Club graph~\cite{zachary1977information, chintalapudi2015survey}. Similar to QuOp, FastRP also works by taking into account neighbors one-hop away then iterates over the set of the neighbors' features to generate node embeddings of a pre-selected size. 

In each experiment in the subsequent subsections, we extract an adjacency matrix for each node by deriving the immediate neighborhood for each node, i.e. a subgraph of the nodes one-hop away, and apply QuOp as described by Algorithm~\ref{alg:QuOp}. We visualize results using a triangular heat-map, since the upper triangle and lower triangle display the same information.

\onecolumngrid

\begin{figure}[!ht]
    \centering
    \begin{tabularx}{\linewidth}{XXXX}
    \begin{subfigure}[b]{.8\linewidth}
    \centering     \includegraphics[height=4.5cm]{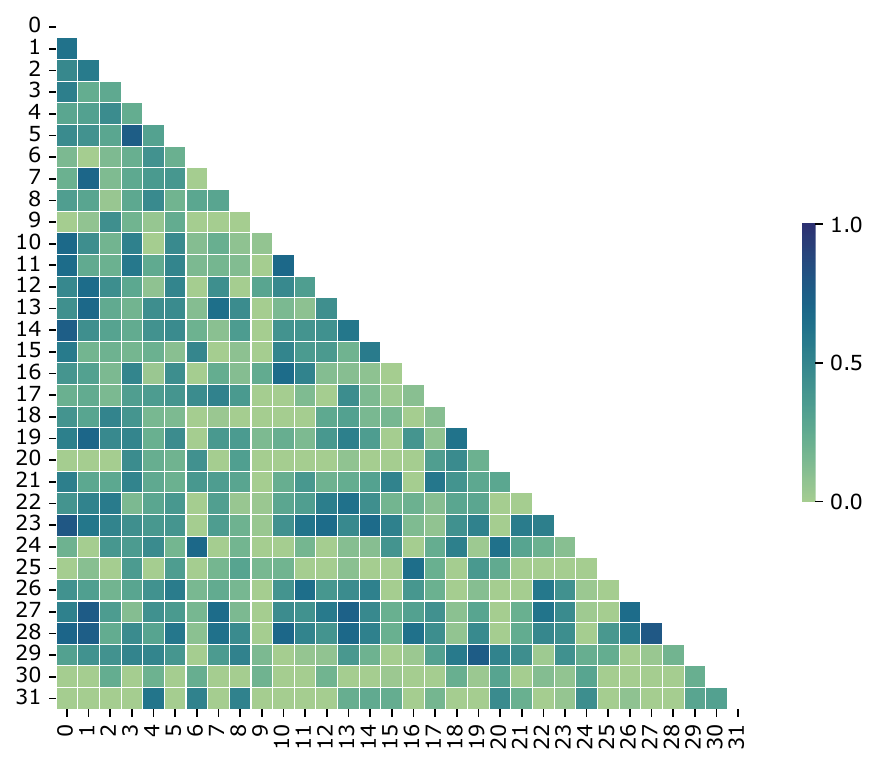}
    \caption{FastRP heat map on the 32 node graph.}
    \end{subfigure} & 
    
    \begin{subfigure}[b]{.8\linewidth}
    \centering\includegraphics[height=4.5cm]{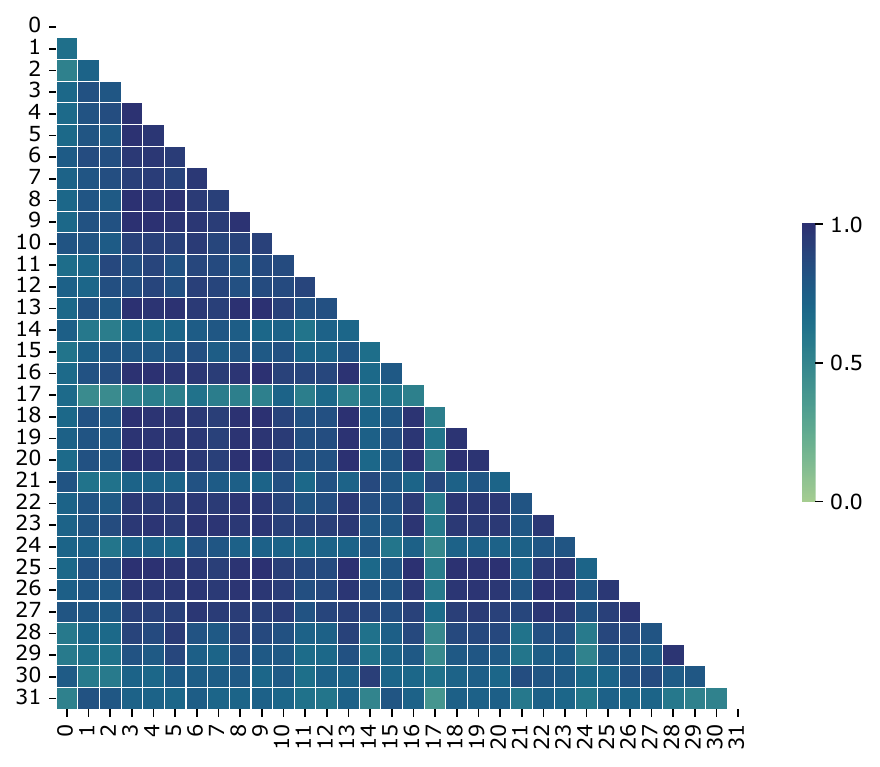}
    \caption{QuOp heat map on the 32 node graph.}
    \end{subfigure}& 
    
    \begin{subfigure}[b]{.8\linewidth}
    \centering     \includegraphics[height=4.5cm]{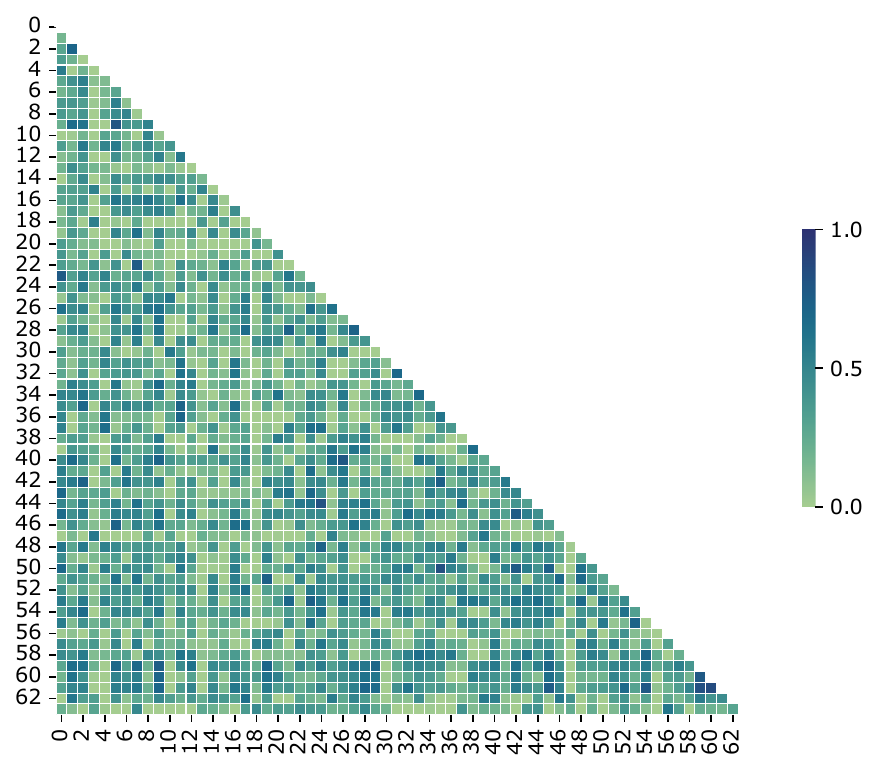}
    \caption{FastRP heat map on the 64 node graph.}
    \end{subfigure}& 
    \begin{subfigure}[b]{.8\linewidth}
    \centering     \includegraphics[height=4.5cm]{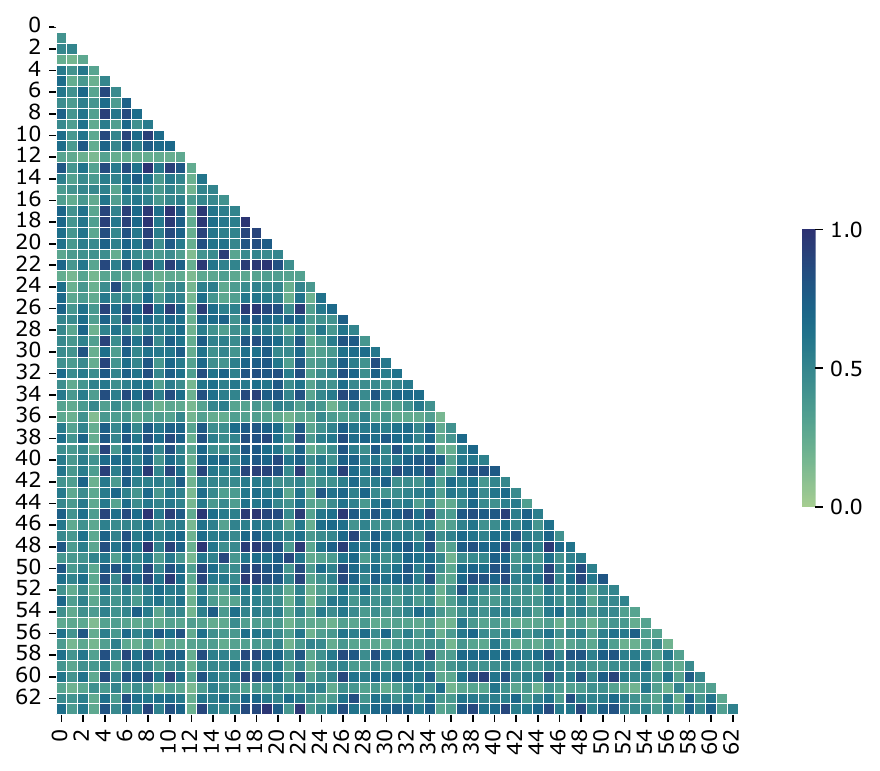}
    \caption{QuOp heat map on the 64 node graph.}
    \end{subfigure}
    \end{tabularx}
    \caption{For the random graphs with 32 nodes and 64 nodes, the heat maps of the node-pairwise similarity scores from FastRP and QuOp.}
    \label{fig:randon-graph-heat-maps}
\end{figure}

\twocolumngrid

In evaluating node embeddings, exploring the concept of node similarity heat maps is crucial. These heat maps visualize the pairwise similarities between nodes in the graph based on their embedding representations. Heat maps provide an intuitive and informative visualization of how closely related or distant nodes are in the embedded space. High similarity values between nodes are depicted with warmer colors, while cooler colors indicate lower similarities.

\begin{figure}[!ht]
    \centering
    \renewcommand{\arraystretch}{1.1}%
    \begin{tabularx}{\linewidth}{lX}
    \includegraphics[height=3.5cm]{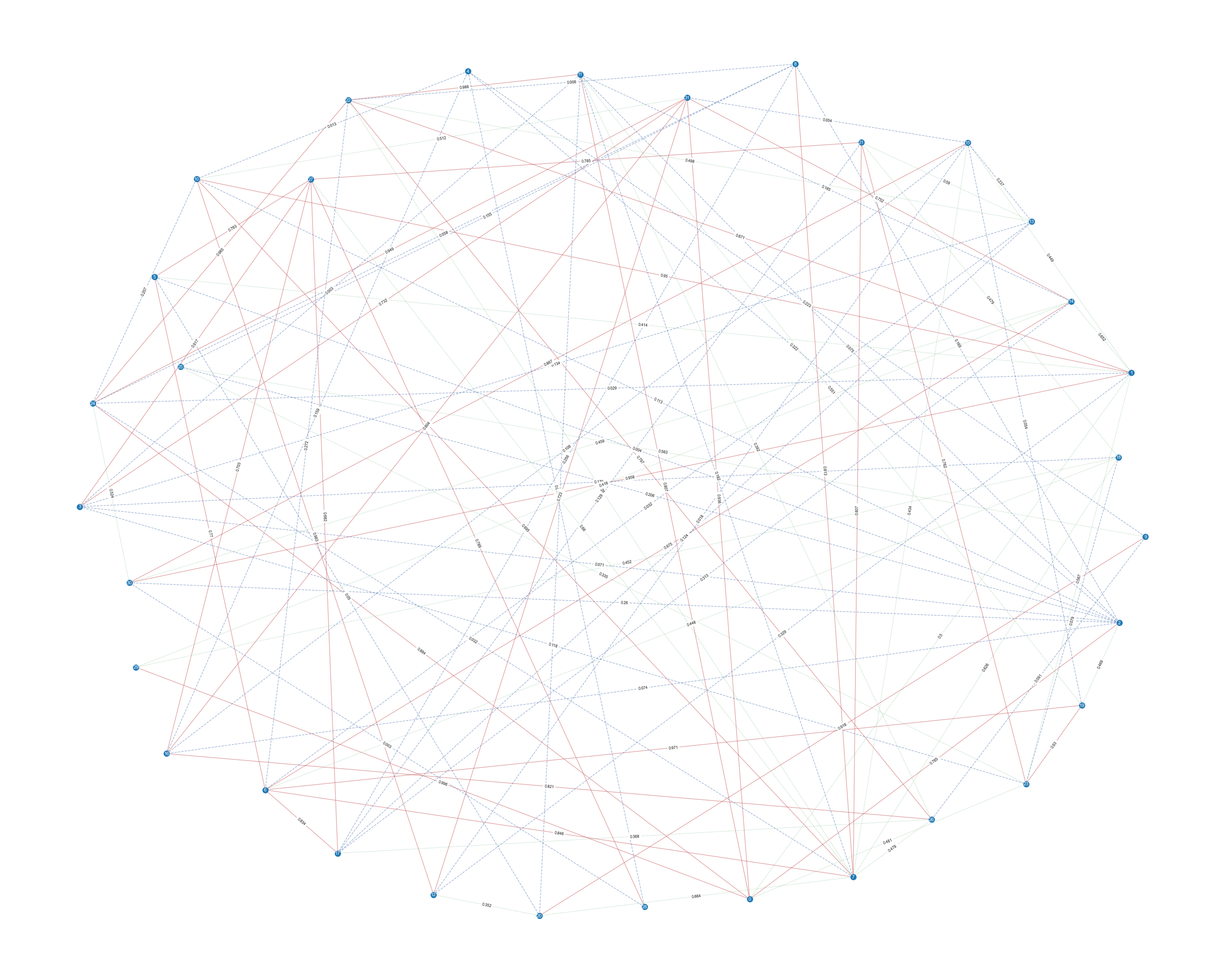} 
    & \includegraphics[height=3.5cm]{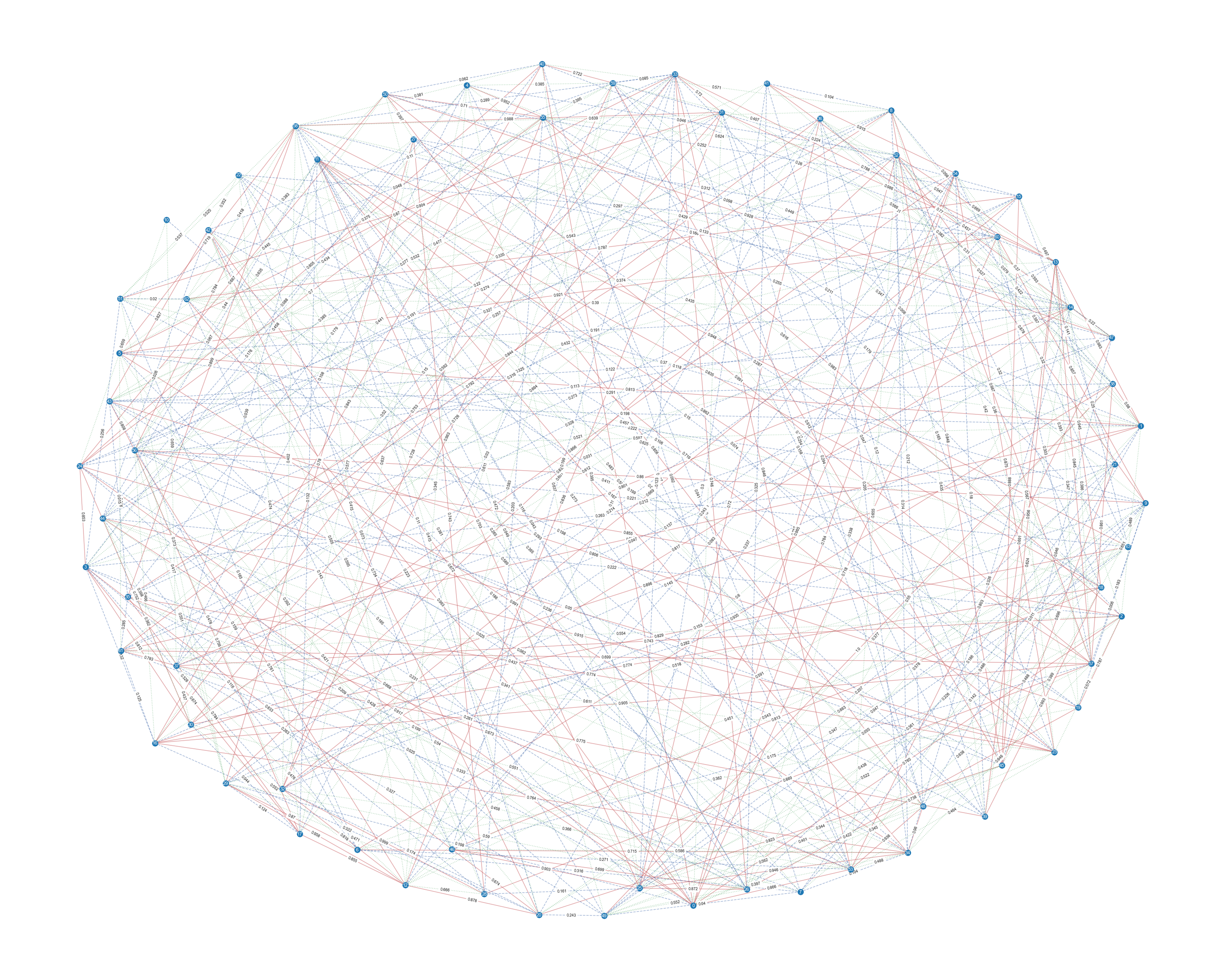}\\
    \multicolumn{1}{c}{(a) Visual of $32$ node graph} 
    & \multicolumn{1}{c}{(e) Visual of $64$ node graph }\\
    \includegraphics[height=3.5cm]{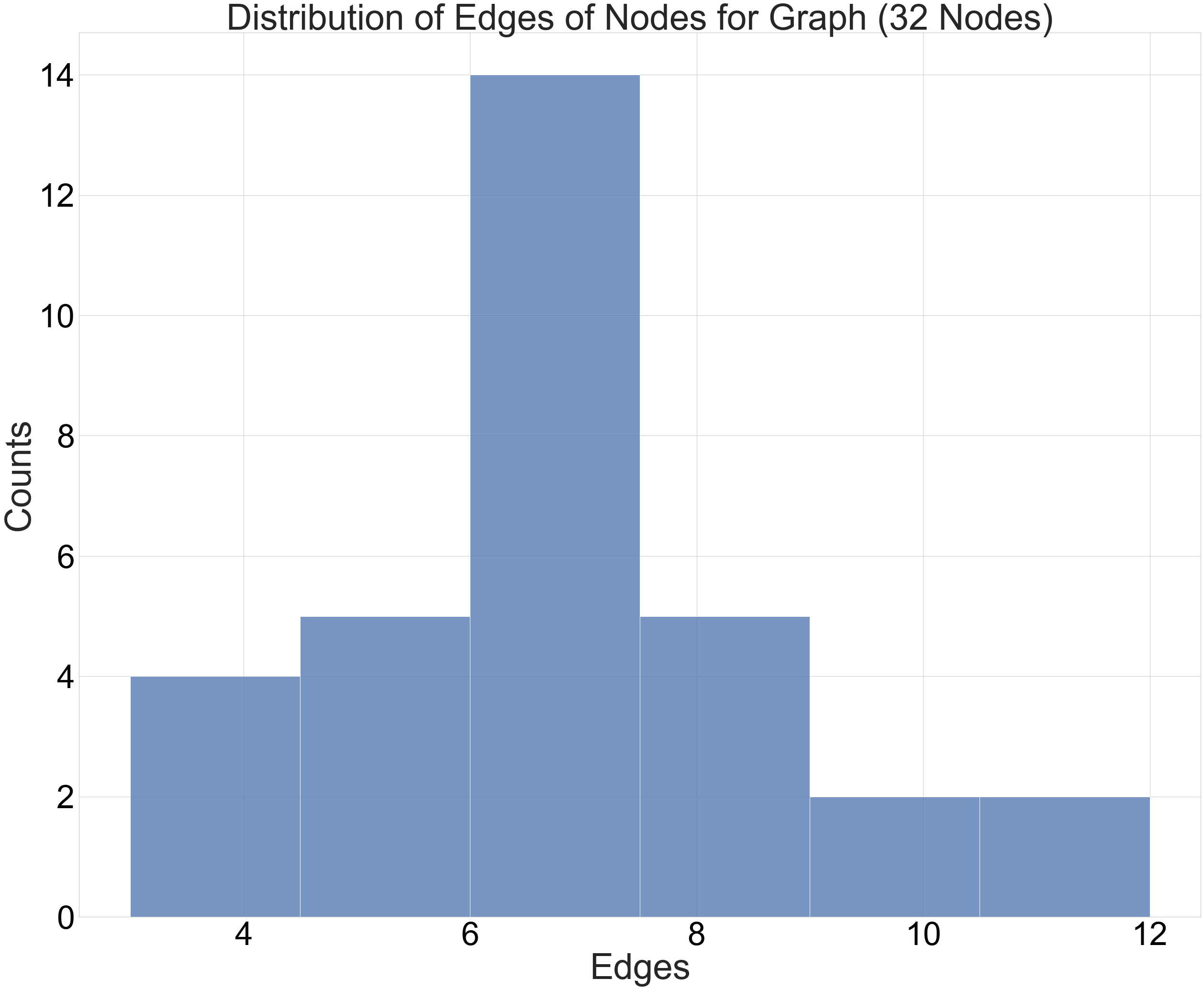} 
    & \includegraphics[height=3.5cm]{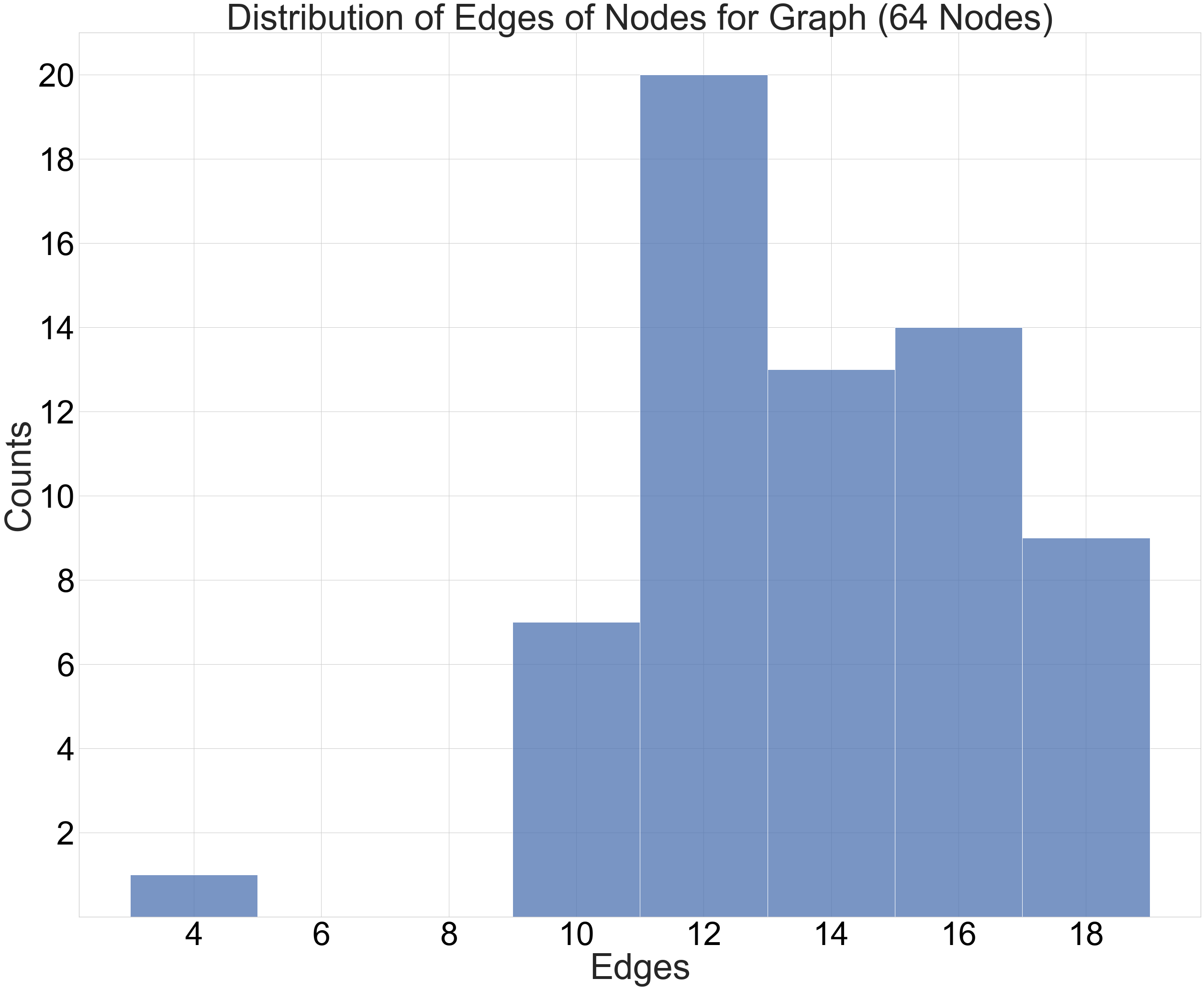}\\
    \multicolumn{1}{c}{(b) Edge count per node.} 
    & \multicolumn{1}{c}{(f) Edge count per node.}\\ 
    \includegraphics[height=3.5cm]{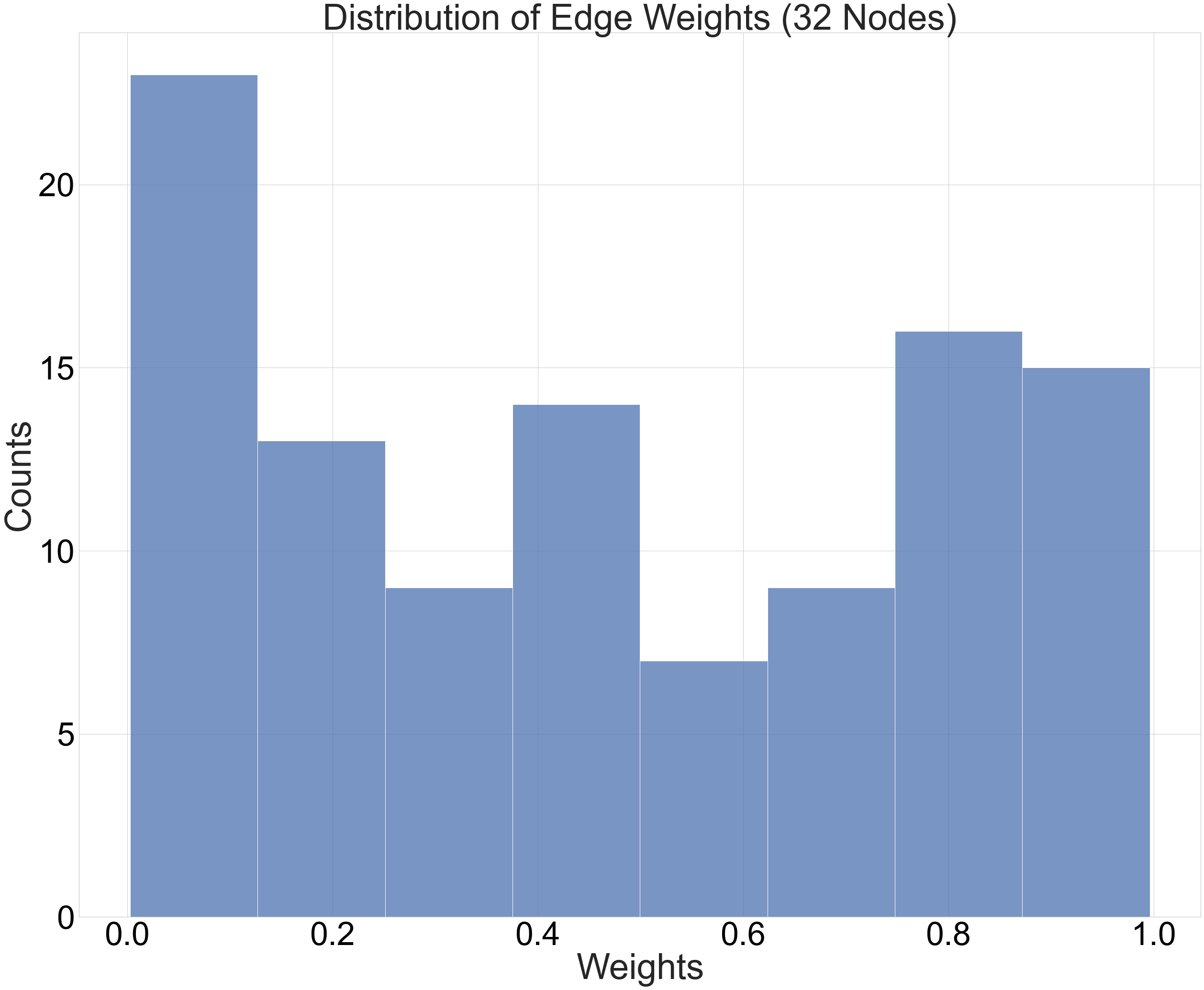} 
    & \includegraphics[height=3.5cm]{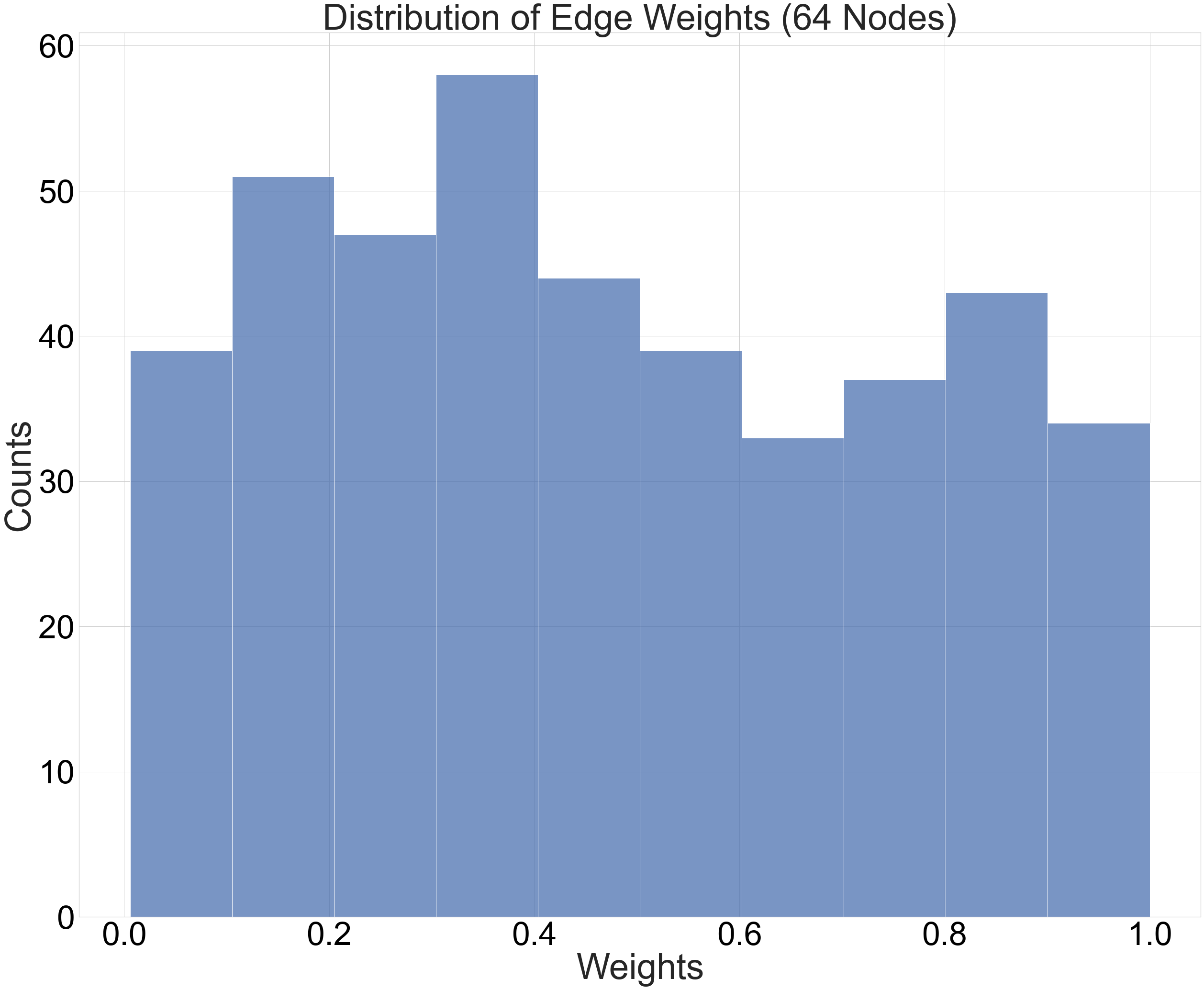}\\
    \multicolumn{1}{c}{(c) Weights on the edges.} 
    & \multicolumn{1}{c}{(g) Weights on the edges.}\\ 
    \end{tabularx}
    \caption{Visualization and distributions of the edges and weights of the random graph with 32 nodes, and random graph with 64 nodes. The column on the left hand side is the 32 node graph, and the column on the right hand side is the 64 node graph.}  
    \label{fig:graph-32-64-visual}
\end{figure}

For the experiments on the Karate Club graph and random graphs, node embeddings using FastRP are generated using an open-source graph database engine from Neo4J, an instance of Neo4J's Community Edition; the enterprise edition also suffices. The Karate Club graph is exported from NetworkX to a file. The graph is then imported into Neo4J. A graph projection for the whole graph is selected and apply the FastRP node embeddings operation on the graph. For ease of use and repeatability, we use Scikit-Learn~\cite{scikit} library implementation of cosine similarity, a normalized dot product of the vectors for nodes $(v_i,v_j)$, to compute similarity scores between pairs.

To further display the potential of QuOp of capturing latent information, we compare QuOp against pre-trained GloVe embeddings \cite{pennington2014glove} with randomly chosen words and the respective corpus. GloVe was selected since it is well-known, a method similar to Node2Vec, and the pre-trained vectors can be adjusted to relative information for a fair comparison, as a full subgraph of one-hop of a word would create too large of an adjacency matrix to run on simulated qubits. Given the intricacy of the process, the full description is in Section \ref{subsec:glove}. 

\subsubsection{Randomly Generated Weighted Graphs}\label{subsec:rando}

To test the initial performance of the ability of QuOp to measure the information distance between nodes in a graph, we start with two graph graphs, one with 32 nodes and the other with 64 nodes, generated at random using an Erdos-Renyi function in the NetworkX Python package \cite{hagberg2008exploring}, and the weights for the edges were uniformly generated with values residing in the unit interval, $(0,1)$; see Figure~\ref{fig:graph-32-64-visual} for an overview. The purpose of the random graphs was to create simple enough graphs which are able to be ran in a simulated quantum environment, while still displaying complexity that would require deeper methods to derive node similarity. 

For the $32$ node graph, a one-hop adjacency matrix was calculated for each node and the QuOp algorithm was ran on simulated qubits with the IBM Qiskit package \cite{aleksandrowicz2019qiskit}, applying $1000$ shots per inner product. For the $64$ node graph, a one-hop adjacency matrix was calculated for each node with $1500$ shots per inner product.

\begin{figure}[!ht]
    \centering
    \renewcommand{\arraystretch}{1.1}%
    \begin{tabularx}{\linewidth}{lX}
    \includegraphics[height=3cm]{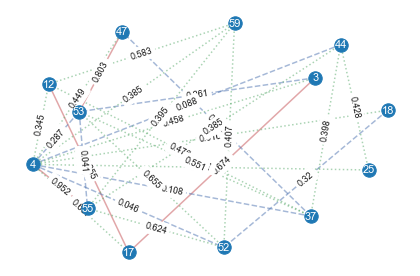} 
    & \includegraphics[height=3cm]{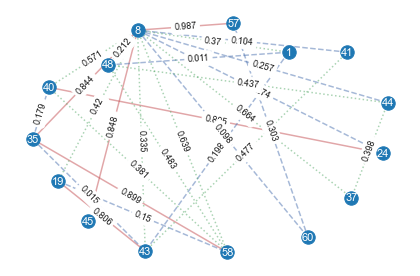}\\
    \multicolumn{1}{c}{(a) Subgraph of node $4$} 
    & \multicolumn{1}{c}{(b) Subgraph of node $8$ }\\
    \includegraphics[height=3cm]{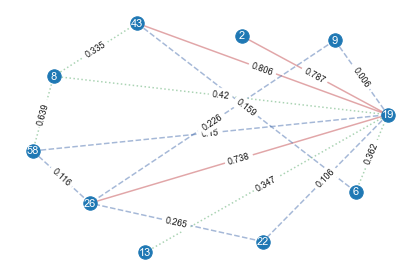} 
    & \includegraphics[height=3cm]{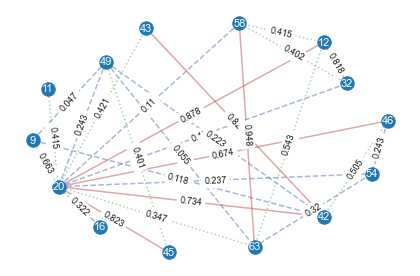}\\
    \multicolumn{1}{c}{(c) Subgraph of node $19$} 
    & \multicolumn{1}{c}{(d) Subgraph of node $20$ }\\ 
    \includegraphics[height=3cm]{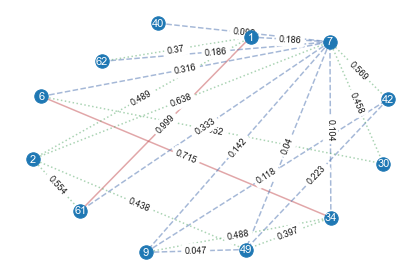} 
    & \includegraphics[height=3cm]{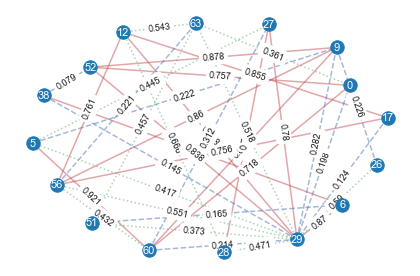}\\
    \multicolumn{1}{c}{(e) Subgraph of node $7$} 
    & \multicolumn{1}{c}{(f) Subgraph of node $29$ }\\ 
    \includegraphics[height=3cm]{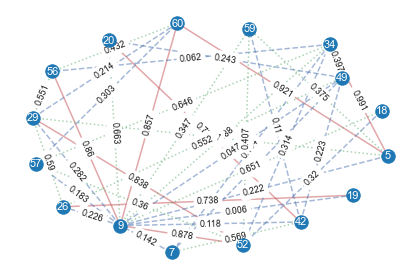} 
    & \includegraphics[height=3cm]{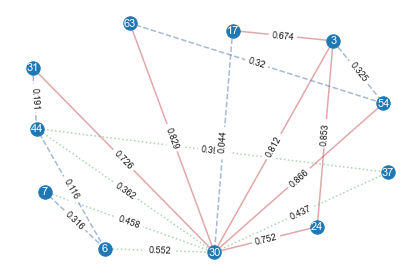}\\
    \multicolumn{1}{c}{(g) Subgraph of node $9$} 
    & \multicolumn{1}{c}{(h) Subgraph of node $30$ }\\ 
    \end{tabularx}
    \caption{Subgraphs of the graph with 32 nodes. The similar nodes are the pairs $(4,8)$ and $(19,20)$, and the dissimilar pairs of nodes are $(7,29)$ and $(9,30)$. }  
    \label{fig:graph-32-nodes-compare}
\end{figure}

\begin{figure}[!ht]
    \centering
    \renewcommand{\arraystretch}{1.1}%
    \begin{tabularx}{\linewidth}{lX}
    \includegraphics[height=3cm]{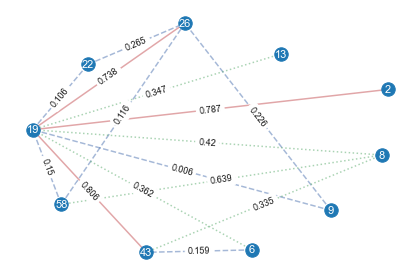} 
    & \includegraphics[height=3cm]{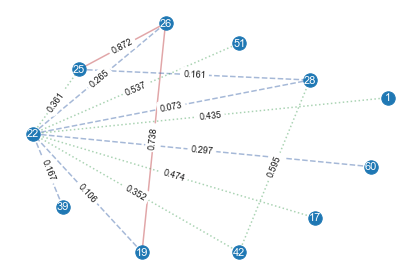}\\
    \multicolumn{1}{c}{(a) Subgraph of node $19$} 
    & \multicolumn{1}{c}{(b) Subgraph of node $22$ }\\
    \includegraphics[height=3cm]{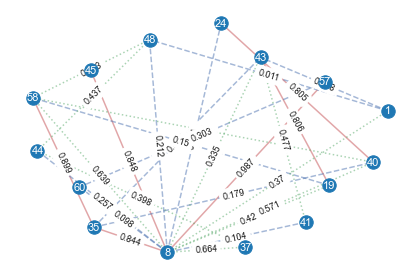} 
    & \includegraphics[height=3cm]{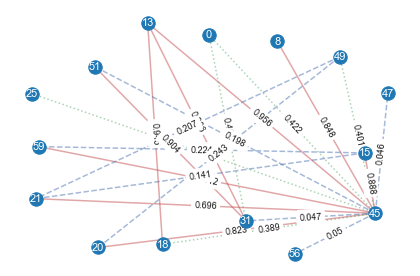}\\
    \multicolumn{1}{c}{(c) Subgraph of node $8$} 
    & \multicolumn{1}{c}{(d) Subgraph of node $45$ }\\ 
    \includegraphics[height=3cm]{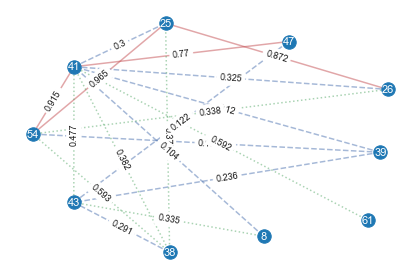} 
    & \includegraphics[height=3cm]{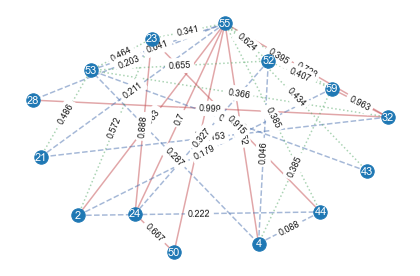}\\
    \multicolumn{1}{c}{(e) Subgraph of node $41$} 
    & \multicolumn{1}{c}{(f) Subgraph of node $55$ }\\ 
    \includegraphics[height=3cm]{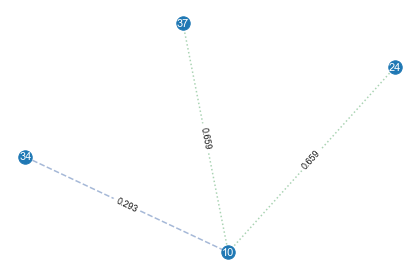} 
    & \includegraphics[height=3cm]{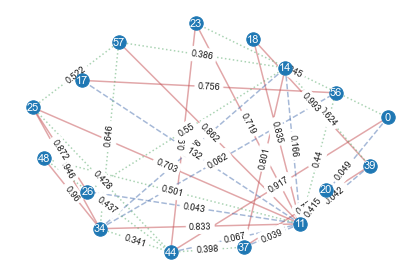}\\
    \multicolumn{1}{c}{(g) Subgraph of node $10$} 
    & \multicolumn{1}{c}{(h) Subgraph of node $11$ }\\ 
    \end{tabularx}
    \caption{Subgraphs of the graph with 64 nodes. The similar nodes are the pairs $(19,22)$ and $(8,45)$, and the dissimilar pairs of nodes are $(41,55)$ and $(10,11)$. }  
    \label{fig:graph-64-nodes-compare}
\end{figure}

For FastRP, we leveraged the Neo4j graph database to store, query, and compute these embeddings at the node-level; for the experiment, embeddings of $16$ dimensions was selected. 

Figure \ref{fig:randon-graph-heat-maps} displays the heat maps for QuOp and FastRP. The heat maps for both graphs display that FastRP does not identify as many similar node pairs that what was identified with QuOp.  However, as information is aggregated in the heat maps, pairs of nodes displaying a similar topology and dissimilar topology were picked to show discrepancies between the scores. The results of this comparison are in Table~\ref{table:info-distance-fastrp}, and a visual inspection of the node pairs are given in Figures \ref{fig:graph-32-nodes-compare} and \ref{fig:graph-64-nodes-compare}. One may observe that Table~\ref{table:info-distance-fastrp} shows that for both graphs QuOp produces consistent results, while FastRP lacks consistency. Hence, QuOp is able to capture the latent information contained within the topology of a node in a better fashion than FastRP.   
 
 To address the scaling, FastRP implemented by Neo4J is fast and scalable and is one widely used~\cite{chen2019fast}. However, as the size of the random graph grows, the similarity between a fixed node pair decreases; this decrease is due to sparseness as the number of nodes and edges increases. On the contrary, QuOp is able to handle the latent information within the topology of a node as the size of the graph increase. The current implementation does not scale well, as we observed limitations when running on conventional hardware. Particularly, a random graph of order 128 nodes was derived, and node comparison of one-hop topologies failed to converge. This limitation may be explored in the future using additional hardware resources, such as leveraging graphic processing units (GPUs).

% \begin{figure}[!ht]
%     \centering
%     \begin{subfigure}[b]{.5\textwidth}
%     \centering     \includegraphics[height=5cm]{figures/random_graph_32node_fastrp.pdf}
%     \caption{32 nodes}
%     \end{subfigure}
%     \begin{subfigure}[b]{.5\textwidth}
%     \centering     \includegraphics[height=5cm]{figures/random_graph_64node_fastrp.pdf}
%     \caption{64 nodes}
%     \end{subfigure}
%     \caption{For the random graphs with 32 nodes and 64 nodes, the heatmaps of the node-pairwise results from FastRP and QuOp. The first row are the results of the 32 node graph, and the second row are the results of the 64 node graph.}
%     \label{fig:randon-graph-heat-maps-quop}
% \end{figure}

\subsubsection{Karate Club Graph} \label{subsec:k-c}

 For a real-world weighted and undirected graph, we selected the Karate Club graph, a university-based Karate Club that separates in two factions or two organizations. The purpose is to examine how QuOp handles the topology of a real-world toy-graph. The graph has node attributes that pair each node to one of two club groups, either ``Mr. Hi" or ``Officer". 
 
 We used NetworkX to generate the Karate Club graph. The node pairs come pre-weighted, and the edges were re-weighted the by dividing by the maximum weight, which is $6$, and multiplying by $\pi$. This adjusted graph is then processed through NEO4J to run FastRP. 

Results for QuOp and FastRP on Karate Club graphs, in Table~\ref{table:karateclub}, show a consistent trend in similarity scores for each node. While scales differ from the first to the second pair, the score trends downwards, with lower similarity. From the second to the third pair, the similarity trends downward. Scores for the first and fifth do differ significantly. With QuOp, pair $(9,18)$ is more similar in contrast to the score obtained from FastRP, and the same is true for the first pair. A closer examination of the topology might shed light on this difference. Pairs $(9,18)$ share the same high degree neighboring node, which signals relative similarity. 

On the other hand, the score for pair $(11,32)$ is high for both methods when the topology suggests otherwise. Node 11 has a degree of 1, whereas 32 has a degree above 10. Each node's one-hop neighbors are different, which would imply lower similarity. 

\begin{table}[!ht]
    \centering 
    \footnotesize
    \caption{Pairwise similarity across a selection of node pairs in Karate Club graph}
    \begin{tikzpicture}
    \node (table) [inner sep=1pt] {
    \renewcommand{\arraystretch}{1.1}%
    \setlength{\tabcolsep}{4pt}

    \begin{tabularx}{.97\linewidth}{lXX}
    % &&\multicolumn{2}{c}{\bfseries{Number of Nodes}}\\
    % \bfseries{\underline{Graph Node Size}} & 
    \bfseries{\underline{Node Pairs}} & \bfseries{\underline{FastRP}} & \bfseries{\underline{QuOp}} \\%& \bfseries{\underline{128}}\\ 
    % \multirow{4}{*}{32} innerprods[32,33], innerprods[0,32], innerprods[0,11], innerprods[11,32],innerprods[11,17]

    (32,31) &  .999     & .318 \\%& 1\\
    (0,32)  &  .410     & .210 \\%& 0.483\\
    (0,11)  &  .397     & .028\\%& 0.333\\
    (11,32) &  .833     & .222 \\%& 0.307\\\hline
    (9,18)  &  .638     & .027 \\%& 0.307\\\hline
    (16,32) & .041      & .102 \\%& 0.307\\\hline
    % \hline
    % (19,22) &  .392    & .855 \\%& -\\
    % (9,30)  &  -.071   & .816 \\%& -\\
    % (41,55) &  .490    & .038 \\%& -\\
    % (10,11) &  .390    & .141 \\%& -\\
    
    \end{tabularx}
    
    };
    \draw [rounded corners=.5em] (table.north west) rectangle (table.south east);
    \end{tikzpicture}
    \label{table:karateclub}
    
\end{table}

\onecolumngrid

\begin{figure}[!ht]
    \centering
    \renewcommand{\arraystretch}{1.1}%
    \begin{tabularx}{\linewidth}{XXX}
    \includegraphics[width = \linewidth]{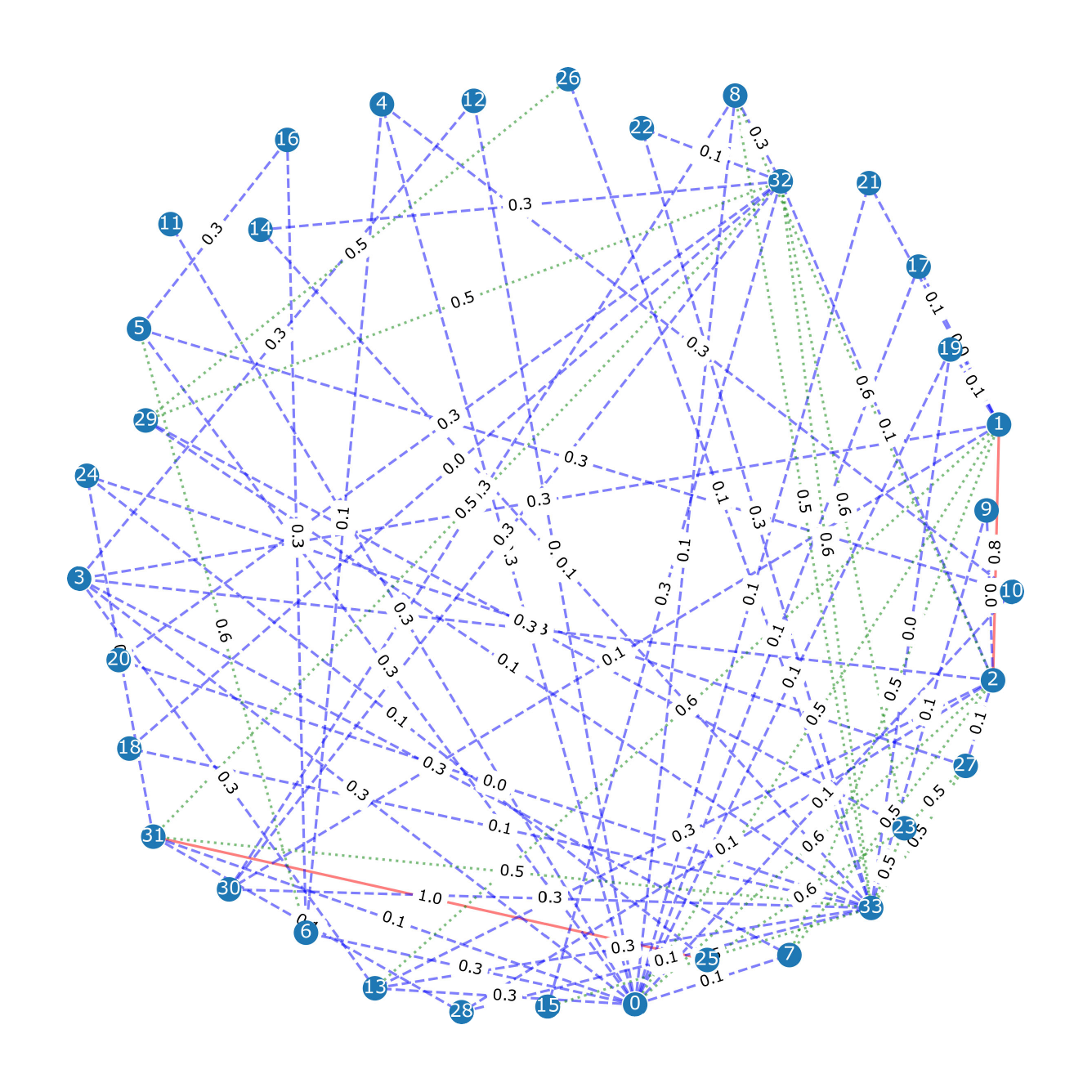} 
    &\includegraphics[width = \linewidth]{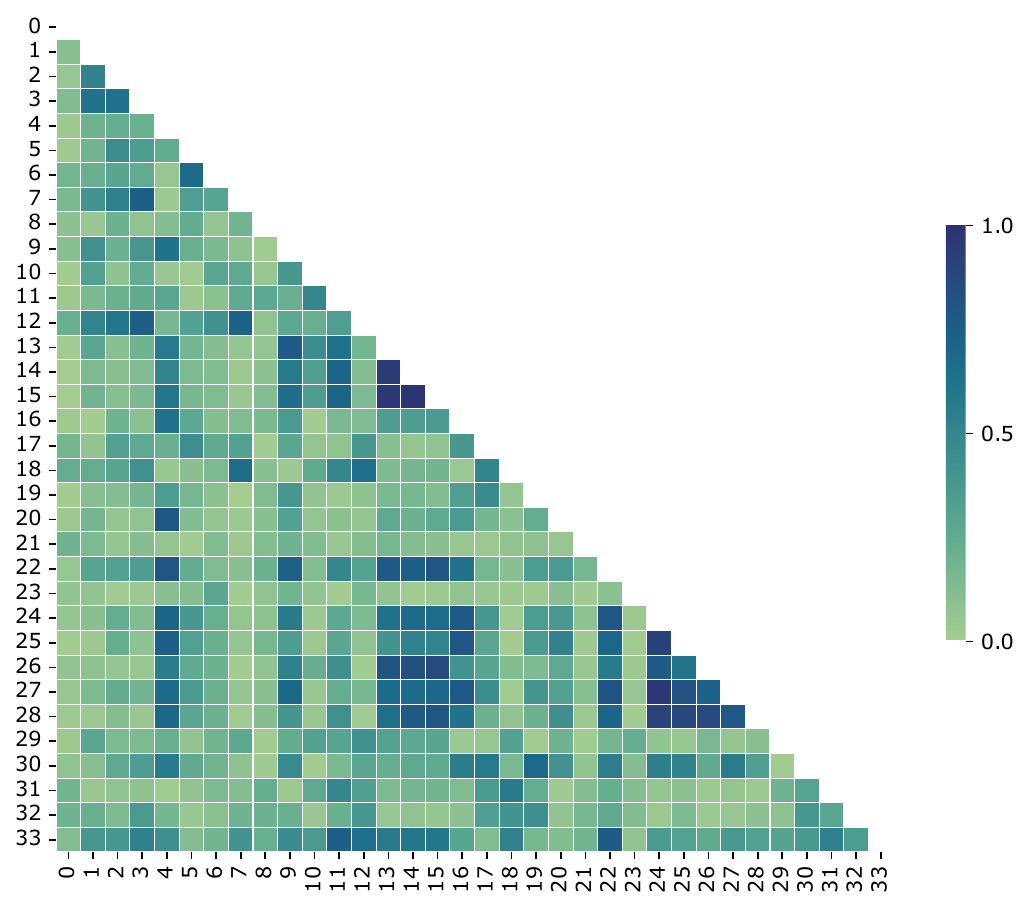} 
    &
    \includegraphics[width = \linewidth]{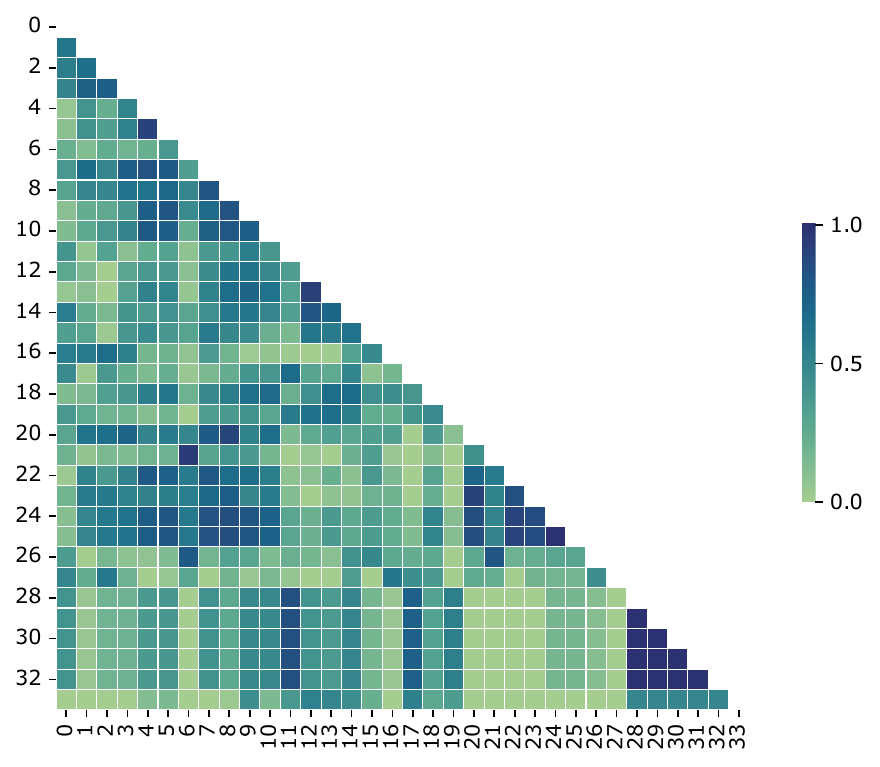}\\
    {(a) Graph topology with edge weights representing normalized interactions between members internal and external to a group.} 
    & {(b) Heat map of inner product values from QuOp, calculated with 1000 shots.}& {(c) Heat map of the cosine similarity from FastRP embeddings.}\\
    \end{tabularx}
    \caption{Zachary's Karate Club benchmark graph~\cite{zachary1977information} results.}  
    \label{fig:karate-club}
\end{figure}
\twocolumngrid

\subsubsection{GloVe Word Embeddings}\label{subsec:glove}

We examine the effectiveness of QuOp on word-based graphs. For this experiment, we utilize the widely recognized GloVe unsupervised learning algorithm developed by Pennington et al.~\cite{pennington2014glove}. We used word-to-word co-occurrence statistics derived from a GloVe embedding algorithm and build graphs using word tokens from the Wikimedia English benchmark Semantic Textual Similarity (STS)\cite{stsbenchmark}. We select sentences from the training text to create a vocabulary following the recommended procedure in GloVe. After processing the tokens, we set a minimum frequency of five appearances per word in the dataset. This approach results in final vocabularies of sizes 32 and 64. We trained the embeddings using the GloVe algorithm while specifying a parameter for the resulting vector's size or dimension, which we set to $50$ in our experiments. 

For QuOp, we used the GloVe word-to-word co-occurrence graph to create a co-occurrence weighted matrix. Particularly, in our experiments, we utilized the frequency value of word co-occurrence to form a weighted adjacency matrix. This matrix was used with the NetworkX package to construct a weighted but undirected graph. The frequency of word co-occurrence was normalized to create a weighted edge list. This list is passed to NetworkX API to generate graphs with weights ranging between zero and one. A graph transform yields adjacency matrices for the nodes without altering the graph's structure.
Figure~\ref{fig:textvecs} highlights word tokens for 32 and 64. We display the first set of word tokens on the left and an additional 32 on the right in their relative positions in two-dimensional space. The 64-node graph extends the 32-node graph by processing extra sentences from the GloVe training documents, allowing the graph to expand with new, unique words.
The 32-word graph's vocabulary lacks enough information in this two-dimensional space to suggest which words are more likely to be found together in a typical sentence. However, when the vocabulary size doubles using GloVe's training documents, word clusters emerge, indicating that these words often appear near each other in English sentences.
Table~\ref{table:info-distance-glove} compares the inner products of QuOp for measuring node pair similarity against GloVe's algorithm, which calculates the similarity between selected word pairs. It should be noted that GloVe does not distinguish between word types (nouns, verbs, etc.), only if a word is next to another word, and tends to be biased. This additional data shows that QuOp consistently outperforms GloVe embeddings. However, as noted in Section \ref{subsec:rando}, running QuOp on standard hardware has limitations, as the runtime significantly increases for vocabularies larger than 64.

\begin{figure}[!ht]
    \centering
    \includegraphics[width=1.1\linewidth]{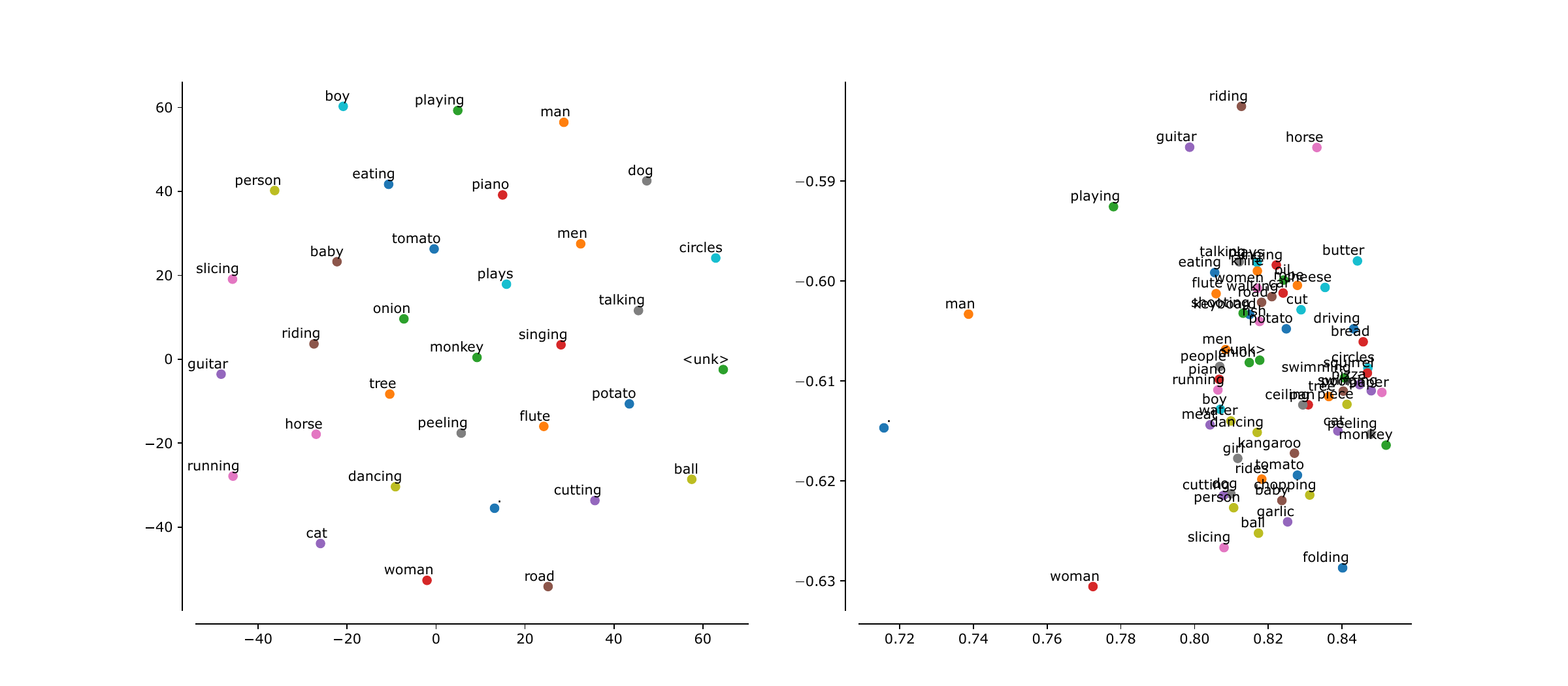}
    \caption{Vector embeddings for 32 and 64 word vocabularies in 2D space.}
    \label{fig:textvecs}
\end{figure}

\begin{table}[ht]
    \centering
    \footnotesize
    \caption{Pairwise node similarity for word-to-word co-occurrence. Similarity scores at ends of the range for QuOp. For co-occurrence graphs, pairs of nodes are selected and compared relative to node similarity given pre-trained vectors.}
    \begin{tikzpicture}
    \node (table) [inner sep=1pt] {
    \renewcommand{\arraystretch}{1.1}%
    \setlength{\tabcolsep}{4pt}

    \begin{tabularx}{.97\linewidth}{llXX}
    &&\multicolumn{2}{c}{\bfseries{Number of Nodes}}\\
    \bfseries{\underline{Vector Emb.}} & \bfseries{\underline{Node Pairs}} & \bfseries{\underline{32}} & \bfseries{\underline{64}}\\
    % & \bfseries{\underline{128}}\\ 
    \multirow{4}{*}{GloVe} & man,playing & 0.7696 & 0.8893 \\%& 0.9683 \\
    & woman,playing & -0.2334 & 0.2418 \\%& 0.7249 \\
    & man,guitar    & 0.5085  & 0.7833 \\%& 0.9631 \\
    & woman,guitar  & -0.4187 & 0.0316 \\%& 0.6867 \\\hline
    \multirow{4}{*}{QuOp}  			
    & man,playing   & 0.8493 & 0.7979 \\%& - \\
    & woman,playing & 0.9279& 0.8907\\%& - \\
    & man,guitar    & 0.9579 & 0.5584\\%& - \\
    & woman,guitar  & 0.9873& 0.5705\\%& - \\
    \end{tabularx}
    
    };
    \draw [rounded corners=.5em] (table.north west) rectangle (table.south east);
    \end{tikzpicture}
    \label{table:info-distance-glove}
    \end{table}

\section{Graphs, Lie Algebras, and Lie Groups}\label{sec:lie}
Lie groups and Lie algebras have an intricate relationship with quantum mechanics \cite{wigner2012group}. As such, there is a natural and beautiful connection with the QuOp algorithm and Lie matrix groups and algebras. Before going into details about this connection, we will describe Lie algebras and Lie groups in general, working towards establishing a foundation and context. 

\subsection{Background}

In quantum mechanics, the Lie group captures the symmetric group of physical systems, and the Lie algebra is the tangent space of this group, which is centered at the identity element \cite{wigner2012group,varadarajan2013lie,bincer2013lie}. As such, the Lie algebra is the infinitesimal transformation of the Lie group and may be thought of as the generating set of this group. 

For a general Lie matrix groups $G$, the respective Lie algebra is denoted as $\mathfrak{g}$, for which $A \in \mathfrak{g}$ and $t \in \mathbb{R}$ we have $e^{tA} \in G$ \cite{hall2013lie,bincer2013lie}. The relationship between the algebra and group is known as the \define{exponential map} (Definition 3.18 in Hall \cite{hall2013lie}). 

This structure gives rise to the following (Theorem 3.20 in Hall \cite{hall2013lie}):
\begin{enumerate}[(i)]
    \item $AXA^{-1} \in \mathfrak{g}$ for $X \in \mathfrak{g}$ and $A \in G$;
    \item $tX \in \mathfrak{g}$ for $t\in \mathbb{R}$;
    \item $X + Y \in \mathfrak{g}$;
    \item and $XY - YX \in \mathfrak{g}$.
\end{enumerate}
In fact, the bracket $[X,Y] = XY - YX$ is called the \define{Lie bracket} or the \define{commutator}, and arises natural from the derivative 
$$
\displaystyle \lim_{h \to 0} \frac{ e^{hX} Y e^{-hX} - Y }{h} = XY - YX.
$$ 
The commutator holds the \define{Jacobi identity} 
$$
\big[X,[Y,Z] \big] = \big[ [X,Y],Z \big] + \big[ Y,[X,Z] \big]
$$ 
and the \define{derivative property} 
$$
\big[ [X,Y],Z] \big] + \big[ [Y,Z],X \big] + \big[ [Z,X],Y \big] = 0
$$ (see Bincer for further information \cite{bincer2013lie}).

As laid out in Bincer \cite{bincer2013lie}, the characteristics of the Lie algebra yields a vector space, and this vector space is decomposable into basis vectors $\{X_a\}_{a \in \{1,\dots, d \} }$.  For $d$ parameters $x=(x_1, \ldots, x_d)$ any element in $G$ sufficiently near the identity element can be expressed as $\displaystyle \exp\Big( \sum_{a=1}^d ix_a X_a \Big)$. The set $\{X_a\}_{a \in \{1,\dots, d \} }$ is denoted as the \define{infinitesimal generators}. 

With the infinitesimal generators, the commutator plays an essential role in Lie algebras. For instance, for any two bases vectors $X_a$ and $X_b$, there exists a basis vector $X_c$ with  $[X_a,X_b] = ig_{ab}^{c} X_c$, where $g_{ab}^c$ are real numbers (since the matrices are Hermitian) and hold the equality $g_{ab}^c = -g_{ba}^c$ (from the commutator). The constants $g_{ab}^c$ are called \define{structure constants}. The structure constants are used to construct matrices, $T_a$, where $(T_a)_{bc} = -ig_{ab}^c$. These matrices are called the adjoint representation \cite{bincer2013lie}.

\subsection{Relation with QuOp}

With this algebraic structure and the QuOp algorithm, observe that for an undirected weighted graph $G=(E,V)$ and number of hops $h$, the set 
$$
   \Big\{ \minus i \cdot A^h_e \Big\}_{e \in E} \subset \Big\{ A \in \mathrm{GL}(N, \mathbb{C}) \ | \ A  = \minus A^{\dagger} \ \& \ \mathrm{Tr}(A) = 0\Big\} 
$$
for some dimension $N$. Recalling the exponential map, which stems from Schr{\"o}dinger's equation, $A \mapsto e^{\minus iA}$ is the reason for the multiplication $\minus i$ constant. This map, not unsurprisingly, is similar to the the representation of matrices near the identity element. For a directed graph applying the $\mathrm{hermAdj}$ function (described by Equation \ref{eq:herm-adj}) yields the same inclusion. 

For a quantum circuit with $n$-wires, the gates of this circuit are operators (or matrices) are contained in the Lie group of special unitary operators, 
$$
\mathrm{SU}(2^n) = \left\{ A \in \mathrm{GL}(N, \mathbb{C}) \ | \ A \cdot A^{\dagger} = I_{N} \ \& \ \det(A) = 1\right\}
$$ 
\cite{nielsen2001quantum}. It is well-known for the Lie group of special unitary operators that the Lie algebra $\mathfrak{su}(N)$ is equal to the set of skew-Hermitian matrices with a trace of zero (Proposition 3.24 in Hall \cite{hall2013lie}). Hence, 
$$
\mathfrak{su}(N) = \Big\{ A \in \mathrm{GL}(N, \mathbb{C}) \ | \ A  = \minus A^{\dagger} \ \& \ \mathrm{Tr}(A) = 0\Big\},
$$ 
and therefore $\Big\{ \minus i \cdot A^h_e \Big\}_{e \in E} \subset \mathfrak{su}(N)$.
\begin{re}
    In general, if $L$ is the discrete Laplacian of a graph then $\mathrm{Tr}(L) \neq 0$, and hence $L \not\in \mathfrak{su}(N)$. However, it should be emphasized that $e^{\minus iL} \in \mathrm{SU}(N)$, and is a viable operator for a quantum circuit, as was briefly noted in quantum random walks, Subsection \ref{subsec:motivation}.
\end{re}
Now, taking the infinitesimal generators into account, a subset $\{X_a\}_{a \in I}$ spans matrices in $\{ \minus i \cdot A^h_e \Big\}_{e \in E}$. Thus, these matrices generate a subalgebra of $\mathfrak{su}(N)$, and therefore, generates a subgroup of $\mathrm{SU}(N)$. 

In applications $N=2^n$ for $n \in \mathbb{N}$, and it would be common $n$ to be very large, which then the matrices  in $\{X_a\}_{a \in I}$ would also be quite large. Iterative applications of the Cartan decomposition \cite{gilmore2006lie} further breaks down the infinitesimal generators $\{X_a\}_{a \in I}$ into, eventually, the Pauli operators and CNOT gates. The decomposition display the generators of the subgroup and yield further insight into the respective flow of information within the graph. 

Of course, this decomposition may not line up with native gates of quantum hardware, but the Cartan decomposition of Lie algebras is able to decompose an arbitrary unitary operator into logical gates \cite{shende2005synthesis,goto2020semisimple,mansky2023near,vartiainen2004efficient,li2013decomposition}. 

For example, the Cartan decomposition algorithm in \cite{Fedoriakagithub} decomposes a special unitary operator into Pauli-$X$ gates, fully-controlled $X$ gates, fully-controlled $R_Y$ gates, and fully controlled $R_Z$ gates. 

\section{Discussion}\label{sec:discussion}

\subsection{Summary}
We derived a variational free method to embed nodes into the space of special unitary operator, $\mathrm{SU}(2^n)$. The method is built on top of the Lie subalgebra, $\mathfrak{su}(2^n)$, generated from local adjacency matrices. Thus, generates the subgroup of $\mathrm{SU}(2^n)$.

We presented a series of experiments designed to assess the efficacy of QuOp for comparing similarity between node pairs in graphs. Our investigation begins with the analysis of two small random graphs, a real-world social network, specifically the renowned Karate Club network~\cite{zachary1977information, chintalapudi2015survey}, and word-to-word co-occurrence graphs. Given the size of the random graphs and the Karate Club graph, we utilized the FastRP method to classically generate node embeddings and calculated similarity of nodes with cosine similarity. Across the three graphs, QuOp consistently displayed scores among similar and dissimilar nodes.   

We then delved into two distinct methodologies for assessing node pair similarity from word co-occurrence. For the word-to-word co-occurrence graph we leverage GloVe embeddings and QuOp. The experiments aim to rigorously evaluate the performance of these methods in capturing intricate relationships within the graph structures, shedding light on their respective strengths and limitations in 
the context of graph-based similarity comparisons. Given the extremely limited data, QuOp, again, consistently displayed the ability to distinguish similar and dissimilar nodes. Of course, as the level of information increases, GloVe's performance will significantly increase. However, due to the hardware limitations, we were unable to run these experiments.

% \begin{table}[ht]
% \centering 
% \caption{Datasets ... }
% \begin{tikzpicture}
%     \node (table) [inner sep=1pt] {
%     \renewcommand{\arraystretch}{1.2}%
%     \begin{tabularx}{.95\linewidth}{rXll}
%     & & \multicolumn{2}{c}{Pairwise Sim}\\
%     % \cmidrule(lr){3-5}
%     Ex & Experiment & QuOp  & FastRP 
%     \\\hline
%     % \cmidrule(lr){2-3}
%     1. & Random graph - Synthetic\\
%     2. & Karate Club - Real-world \\
%     3. & Word Coocurrence (GloVe) - Real-world \\
%     \end{tabularx}
%     };
%     \draw [rounded corners=.5em] (table.north west) rectangle (table.south east);
%     \end{tikzpicture}
%     \end{table}

\subsection{Potential Future Research}

\subsubsection{Advancements of the Embedding}

Graph embeddings, as vectors, when two vectors are compared with the typical cosine similarity or Spearman correlation the values fall in the interval $[\minus 1,1]$. Hence, embeddings have the ability to calculate both the correlation or the anticorrelation. This in contrast to the QuOp algorithm as the similarity scores are in the unit interval, $[0,1]$, and not able to describe that extra complexity of relationship. This shortcoming should be researched further. 

The graph embedding considered is variational free, however, there could be an advantage to adding tunable layers. Of course, the derivation of the layers tend to be ansatz motivated, hence intuition driving the derivation. The root structure of the Lie subalgebra generated from the local adjacency matrices from each node, with a given number of hops, will yield the ansatz of operators to leverage that will keep additional variational embedding layers within the subalgebra, and thus the within the subgroup of operators. Therefore retaining the latent structure of the graph. 

The modeling considered in the research essentially only considers a node from a homogeneous graph. Heterogeneous graphs tend to have a quantitative representation of each node as a manner to distinguish either collection of similar nodes, or individual nodes. Ergo, weighted nodes. Many of the techniques mentioned in Section \ref{subsec:compare} identify a wire for each node, which enables an easy incorporation of weights into a quantum circuit, that is, after the operators and layers have been identified. Since the number of wires are $\log_2$ of the length of the adjacency matrix, this makes the task difficult. One possible method, for each adjacency matrix add the respective matrix, where the entries are the weights of each node in the adjacency matrix, and placed in the respective entries in the adjacency matrix. Call this matrix the \textit{weighted node matrix}. This would ensure the weighted node matrix is Hermitian with zero trace. Therefore, these new matrices would generate an extended Lie subalgebra. This could complicate the structure of the node representation, and an implementation can utilize the Trotter-Suzuki formula \cite{suzuki1976generalized} into two operators, the adjacency generated operator and the weighted matrix operator.

\subsubsection{Different Graph Representations}
On the surface, graphs are simple objects that, coupled with a visual representation, are an intuitive and easy to understand structure of objects. However, the rich literature on classical and machine learning methods on graphs have displayed an intricate latent space of information embedded within a graph. In general, mathematics has a deep wealth of tools to extract latent information from objects. Below is a description of a potential path to extract information in an explainable and mathematically rigorous manner. Since this is only a potential path, an overarching outline is given.      

\paragraph{Brief Overview of Category Theory}\label{subsec:quivers-cat-overview}

Category theory is a relatively new, but well-established, branch of mathematics that uses maps to describe the relationship between structures that, at first glance, has seemingly different characteristics. 

A collections of objects and collection of respective morphisms between these object is called a \define{category}. A map between categories is called a \define{functor}, and a map between functors is defined as a \define{natural transformation} \cite{mac2013categories,riehl2017category}. Functors and respect natural transformations extract information from collections by mapping a category of interest to another category where it is easier to analyze.

For an illuminating example, consider the category of groups, denoted as \define{Grp}. This category is very intuitive as the objects are a collection of elements with a particular operation that make up a group and the morphisms are homomorphism (see NcWeeny \cite{mcweeny2002symmetry} for a background in group theory). To display a functor, we will connect \define{Grp} to the category of sets, denoted as \define{Set}. The category Set consists of sets as objects and total functions (defined on every element is its domain) as morphisms. There is a fairly intuitive functor between Grp and Set which takes the collection of elements with some structure to just a collection of elements, and homomorphisms to just a morphism that does not necessarily retain the structure of the group (since it is no longer necessary). This functor is called the forgetful functor, as the binary group operation is not retained, stripping the algebraic structure down to just the elements. See Mac Lane or Riehl 
 for a deeper description \cite{mac2013categories,riehl2017category}, or Bradley \cite{bradley2018applied} for a brief overview. 

\paragraph{Quiver Representations}\label{subsec:quivers-repre} 

To make a graph amenable to representation, it is first broken down into a category while keeping the structure in-tack through a quiver. A quiver then enables a graph to represent a vertex as a vector space and an edge as a linear operator between vector spaces, by first considering a free category. Given the arbitrary nature of a representation, which has the potential to be forgetful of the structure of the graph, maps between representations and, as a consequence, all representations will be considered. From these collective maps and representations a graph will then have a natural structure amenable to categorical theoretic techniques to map to quantum operators. Below a granular description of this technique is given.  

Given a directed graph $G=(V,E)$, a \define{quiver}, denoted as $Q$, is a category with $Q_0$ the set of vertices, $Q_1$ the set of arrows which are the edges, and two morphisms $s,t:Q_1 \to Q_0$ where $s$ sends an edge to its source node and $t$ sends an edge to its target node. A \define{path} is then a sequence of arrows $e_1\ldots e_m$ where $t(e_i) = s(e_{i+1})$ for $1\leq i \leq m-1$ \cite{crawley1992lectures,savage2005finite}. 

To make $Q$ more tangible, each node $v$  will be mapped via $\rho$ to a vector space $V_v$ with basis space $\mathbb{K}$. For an arrow (edge) with  $v \xrightarrow{e} v'$ we have the linear operator $V_v \xrightarrow{\rho(e)} V_{v'}$. Therefore, $Q$, via representation theory, is mapped to a collection spaces and morphisms, where the nodes are mapped to $\mathbb{K}$ vector spaces, and edges to linear homomorphisms between vector spaces \cite{dlab1976indecomposable,deng2008finite}.  

One may see such this representation is contained in the category of vector spaces, denoted as $\mbox{Vect}_{\mathbb{K}}$. This mapping may not necessarily consider the composition of maps, which captures each path in the graph. Adjusting for this shortcoming, we take a functor between the free category of $Q$ (category created by `freely' concatenating morphisms), denoted as $\mbox{FrCat}(Q)$, and the category of vector spaces $\mbox{FrCat}(Q) \xrightarrow{\phi} \mbox{Vect}_{\mathbb{K}}$. Thus, a quiver representation is a functor $\phi$. Since a functor has the potential to create bias we take a natural transformation, $\Psi$, to map between functors. Ergo, $\Psi: \phi \Rightarrow \phi'$. And therefore, a representation of $Q$ with base space $\mathbb{K}$, $\mbox{Rep}_{\mathbb{K}}(Q)$, is isomorphic to this category of functors. Hence, 
$$
    \mbox{Rep}_{\mathbb{K}}(Q) \cong \mbox{Vect}_{\mathbb{K}}^{\mbox{FrCat}(Q)},
$$
where $B^C$ denotes the category of functors $C\xrightarrow{\phi} B$ \cite{derksen2005quiver,schiffler2014quiver}. 

\paragraph{Representations of the Special Unitary 
 Group}\label{subsec:su-repre}

With the derived categorical structure of graphs, the goal is then to map this structure to quantum operators in the group $\mathrm{SU}(2^n)$ for some $n \in \mathbb{N}$. There are two potential paths to obtain this goal. The first is to take the representation of this group, and the second is to decompose the Lie algebra into its Dynkin diagram. 

Recall that a representation for a Lie group $L_G$ is a group homomorphism $\Pi : L_G \to \mathrm{GL}(N,\mathbb{K})$ \cite{hall2013lie}. Hence, a representation may be seen map as a subcategory to within $\mbox{Vect}_{\mathbb{K}}$ where the vector space stays the same dimension. Since the dimension of the vector space has an influence on the representation and it is not clear which dimension is optimal, and representations are not unique, all possible representations need to be considered. This foundation may lead to a derivation of a morphism to $\mbox{Rep}_{\mathbb{K}}(Q)$ in some manner.  

Dynkin diagrams is a method to describe a Lie algebra through the roots and the respective structure of the roots. This description, in turn, describes this structure via a graph \cite{hall2013lie,bincer2013lie}. Thus, the generators of a Lie algebra are identified, analyzed, then mapped to a graph. The root systems for Lie algebras are well-classified \cite{bincer2013lie}. With the Dynkin diagram graph structure one may then utilize the graph representation in Section \ref{subsec:quivers-repre}.

\section{Disclaimer}

About Deloitte: Deloitte refers to one or more of Deloitte Touche Tohmatsu Limited, a UK private company limited by guarantee (“DTTL”), its network of member firms, and their related entities. DTTL and each of its member firms are legally separate and independent entities. DTTL (also referred to as “Deloitte Global”) does not provide services to clients. In the United States, Deloitte refers to one or more of the US member firms of DTTL, their related entities that operate using the “Deloitte” name in the United States and their respective affiliates. Certain services may not be available to attest clients under the rules and regulations of public accounting. Please see  www.deloitte.com/about to learn more about our global network of member firms.

Deloitte provides industry-leading audit, consulting, tax and advisory services to many of the world’s most admired brands, including nearly 90\% of the Fortune 500® and more than 8,500 U.S.-based private companies. At Deloitte, we strive to live our purpose of making an impact that matters by creating trust and confidence in a more equitable society. We leverage our unique blend of business acumen, command of technology, and strategic technology alliances to advise our clients across industries as they  build their future. Deloitte is proud to be part of the largest global professional services network serving our clients in the markets that are most important to them. Bringing more than 175 years of service, our network of member firms spans more than 150 countries and territories. Learn how Deloitte’s approximately 457,000 people worldwide connect for impact at  www.deloitte.com.

This publication contains general information only and Deloitte is not, by means of this [publication or presentation], rendering accounting, business, financial, investment, legal, tax, or other professional advice or services. This [publication or presentation] is not a substitute for such professional advice or services, nor should it be used as a basis for any decision or action that may affect your business. Before making any decision or taking any action that may affect your business, you should consult a qualified professional advisor.
Deloitte shall not be responsible for any loss sustained by any person who relies on this publication. Copyright © 2024 Deloitte Development LLC. All rights reserved.

\bibliography{QEbib}
\end{document}